\documentclass{article}

\usepackage[round,authoryear]{natbib}

\usepackage[utf8]{inputenc}  
\usepackage[T1]{fontenc}     
\usepackage{hyperref}        
\usepackage{url}             
\usepackage{booktabs}        
\usepackage{amsfonts}        
\usepackage{nicefrac}        
\usepackage{microtype}       
\usepackage{xcolor}          
\usepackage{adjustbox}
\usepackage{algorithm}
\usepackage{algorithmic}
\usepackage{amsmath}
\usepackage{graphicx}
\usepackage{subcaption}
\usepackage{float}
\usepackage{amsthm}
\usepackage[margin=1in]{geometry}   
\usepackage{multirow}
\usepackage{makecell}
\usepackage{wrapfig}
\usepackage{arydshln}

\newtheorem{theorem}{Theorem}[section]

\newtheorem{lemma}[theorem]{Lemma}

\newtheorem{assumption}[theorem]{Assumption}

\usepackage{amsmath,amsfonts,bm}

\def\eqref#1{equation~\ref{#1}}

\def\1{\bm{1}}

\def\vtheta{{\bm{\theta}}}

\def\vh{{\bm{h}}}

\def\vx{{\bm{x}}}
\def\vy{{\bm{y}}}

\DeclareMathAlphabet{\mathsfit}{\encodingdefault}{\sfdefault}{m}{sl}
\SetMathAlphabet{\mathsfit}{bold}{\encodingdefault}{\sfdefault}{bx}{n}

\def\gD{{\mathcal{D}}}

\def\gL{{\mathcal{L}}}

\def\gQ{{\mathcal{Q}}}

\newcommand{\ds}{\gD}

\newcommand{\lData}{\vx}                 
\newcommand{\lLable}{\vy}                
\newcommand{\lAct}{\vh}               

\newcommand{\layerNum}{N}
\newcommand{\bNum}{L}

\newcommand{\resF}{F}                 
\newcommand{\loss}{\gL}               
\newcommand{\lW}{\vtheta}                   
\newcommand{\lGrad}{g}                
\newcommand{\lEta}{\eta}              
\newcommand{\lRho}{\rho}              

\title{CommuteProp: Decoupled Training for Communication Bound Split LLM Fine-Tuning}

\providecommand{\name}[1]{\textbf{#1}}
\providecommand{\email}[1]{\texttt{#1}}
\providecommand{\addr}[1]{#1}

\author{\name{Ding Chen}$^{1}$ \email{ding.chen@my.cityu.edu.hk} \\
        \name{Haochen Luo}$^1$ \email{chester.hc.luo@my.cityu.edu.hk} \\
        \name{Chen Liu}$^1$ \email{chen.liu@cityu.edu.hk} \\
        \addr{$^1$ Department of Computer Science, City University of Hong Kong} \\
        83 Tai Chee Ave, Kowloon Tong, Hong Kong, China
        }

\date{}

\begin{document}

\maketitle

\begin{abstract}
Split learning has emerged as a promising paradigm for privacy-preserving LLM fine-tuning, yet its practical deployment is severely hindered by the sequential communication-computation bottleneck. In conventional synchronous pipelines, clients remain idle while waiting for server-side gradients, resulting in substantial training inefficiency. We propose CommuteProp, an asynchronous split-learning algorithm that decouples the training process into two concurrent phases: a cross-block forward-backward pass and an in-block weight update. By overlapping computation with communication, CommuteProp reduces the marginal cycle time from a sum of all stage latencies to the dominant computational bottleneck. We provide a rigorous asynchronous error and convergence analysis. Moreover, we derive an NS preconditioner method based on our analysis to further mitigate staleness-induced noise. Comprehensive experiments indicate that our algorithm achieves substantial throughput gains while maintaining accuracy comparable to synchronous methods; furthermore, it functions as a plug-and-play module that not only accommodates but actively enhances existing privacy enhancement methods.
\end{abstract}

\section{Introduction}
The increasing demand for fine-tuning Large Language Models (LLMs)\citep{vaswani2017attention} on proprietary datasets has encountered a fundamental conflict between model utility and data privacy. In highly regulated sectors such as healthcare and finance, raw data is highly regulated from leaving local custody due to stringent privacy protocols\cite{voigt2017eu}. While Federated Learning (FL) \citep{kairouz2021advances} offers a decentralized alternative, the colossal memory requirements of LLMs often exceed the hardware limits of local clients, rendering full-model local training infeasible. Consequently, Split Learning (SL) \citep{gupta2018distributed} has emerged as a critical paradigm. By partitioning the model across institutional boundaries, SL allows clients to retain sensitive data locally while offloading the memory-intensive computation to a powerful server. Specifically, let the model $\mathcal{F}$ be partitioned into $N$ layers with parameters $\theta_1,\dots,\theta_N$, let $\mathcal{D}$ be the data distribution with samples $(\lData,\lLable)$, and let $\loss_{\text{task}}$ denote the loss function. We have the following mathematical problem:

\begin{equation}
\label{equ:split_format}
\begin{aligned}
\min_{\lW_1,\ldots,\lW_\layerNum} \quad  &\mathbb{E}_{(\lData, \lLable) \sim \ds} \left[ \loss_{\text{task}} \left( \mathcal{F}(\lData; \lW_1, \ldots, \lW_\layerNum), \; \lLable \right) \right] \\
\Theta_{\text{client}} = \{\lW_i \mid i \in \mathcal{I}_{c}\},  \Theta_{\text{server}} &= \{\lW_j \mid j \in \mathcal{I}_{s}\}, \quad \mathcal{I}_{c} \cup \mathcal{I}_{s} = \{1, \ldots, \layerNum\}, \mathcal{I}_{c} \cap \mathcal{I}_{s} = \emptyset
\end{aligned}
\end{equation}

Parameters on the client side where the private data is hosted are included in $\Theta_{\text{client}}$, while bulk part of parameters include in $\Theta_{\text{server}}$ are on the server side that has no access to the private data.
Although SL offers a promising privacy-preserving paradigm, it is severely bottlenecked by substantial communication overhead. In the conventional synchronous execution, the client remains idle during the server-side computation and the subsequent backpropagation of gradients. 
Specifically, let $t_{\text{client}}$ and $t_{\text{server}}$ denote the local processing times for the client and server, respectively, and $t_c$ represent the one-way communication latency. The marginal cycle time for a single batch under this synchronous paradigm is strictly additive $\Delta T_{\text{sync}} = t_{\text{client}} + t_{\text{server}} + 2t_c$. Unlike the centralized settings that $t_c$ can be ignored due to the high-speed interconnects (e.g. NVLink), in SL, $t_c$ indicates that the training process faces sever communication overhead \textbf{in every iteration}.

To address these inefficiencies, we propose \textbf{CommuteProp}, an asynchronous SL algorithm designed to overlap computation with communication. We decompose the training process into two distinct phases: \textit{cross-block (Phase1)} and \textit{in-block (Phase2)} forward-backward passes. This decoupling allows the client to utilize its previously idle cycles to execute in-block operations for subsequent micro-batches. 
Ideally, the client no longer waits for the server's gradient feedback to initiate the next forward pass. Consequently, the steady-state cycle time is theoretically governed by the dominant computational bottleneck $\Delta T_{async} = \max(t_{\text{client}}, t_{\text{server}}, t_c)$. Figure \ref{fig:overall} demonstrates the superior pipeline efficiency of our proposed algorithm compared to the synchronous baseline in SL. 

\begin{figure}[htbp]
\vspace{-0.2cm}
    \centering
    \includegraphics[width=1.0\linewidth]{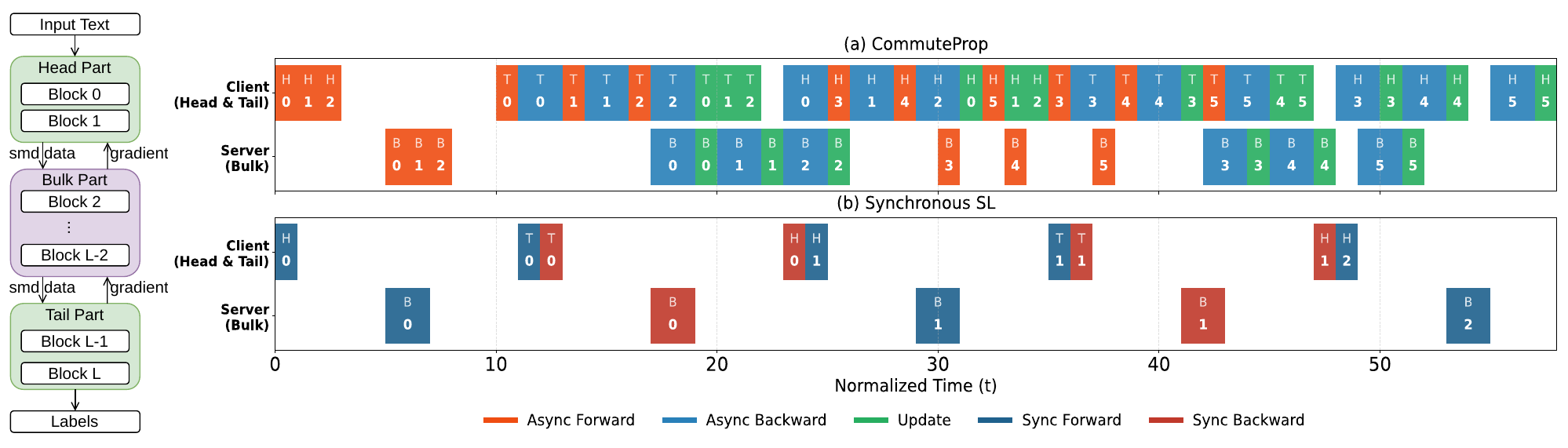}
    \caption{ Illustration of pipeline efficiency. Left: USplit(specific scenario of SL) learning overview. Right: Gantt chart comparison between the proposed algorithm and the synchronous baseline in a USplit learning framework.  H, B, and T denote the head, bulk, and tail parameters, respectively.}
    \label{fig:overall}
\end{figure}

The Gantt chart in Figure \ref{fig:overall} indicates that our asynchronous pipeline fully utilizes client resources, whereas the synchronous baseline suffers from substantial idle time due to sequential waiting. While asynchrony inevitably introduces gradient staleness, we provide a rigorous convergence analysis for CommuteProp. Our analysis reveals that the asynchronous approximation error is bounded by a term proportional to the model’s cross‑block Lipschitz constant — a quantity that prior work \cite{zhang2024transformers} shows to be typically quite small in LLM. Building on this analysis, we further introduce a preconditioner that actively reduces the asynchronous error.

Comprehensive numerical experiments are conducted to evaluate the empirical performance of CommuteProp across various benchmarking tasks. Our results demonstrate that the proposed algorithm significantly improves training throughput while maintaining competitive model accuracy. Moreover, we provide the ablation study, further shows the impact of gradient staleness on convergence stability, providing insights into the trade-off between asynchrony and optimization trajectory. Furthermore, consistent with recent work \cite{liu2025dualguard} that introduces privacy-preserving methods in SL, our experiment indicates that CommuteProp, as a plug-and-play framework, not only accommodates but actively enhances existing privacy enhancement methods.

We summarize the contribution of this paper as follows: (1) We propose CommuteProp, an asynchronous split‑learning algorithm that overlaps computation and communication via a two‑phase pipeline; (2) We propose a rigorous asynchronous error and convergence analysis. We also introduce the NS preconditioner method based on our analysis to reduce the asynchronous error; (3) Comprehensive experiments on NLP benchmarks demonstrate substantial throughput gains, competitive accuracy, and plug‑and‑play compatibility with privacy‑enhancing methods.

\section{Related Works}

\noindent \textbf{SL in LLM Fine-Tuning.} Split Learning (SL) \cite{gupta2018distributed, vepakomma2018split} partitions a model across client and server, allowing the client to offload heavy computation while retaining only a few layers locally. U‑Shaped Split Learning (USL) \cite{gupta2018distributed, lyu2023optimal} further assigns the head and tail to the client, with the trunk delegated to the server. This architecture is well‑suited for autoregressive LLMs. However, recent work \cite{gu2025vflair, chen2024unveiling} shows that SL remains vulnerable to data reconstruction attacks, allowing servers to recover portions of the client’s private data. This has led to defense mechanisms \cite{liu2025dualguard, gu2026differentially} designed to mitigate such risks and strengthen the training process.

\noindent \textbf{Pipeline Parallelism.}
SL is structurally similar to Pipeline Parallelism (PP). In HPC settings, synchronous PP with micro‑batching achieves high throughput and negligible communication overhead \cite{harlap2018pipedream,qi2023zero,liu2024deepseek}. When communication is costly, asynchronous alternatives become preferable, but they suffer from gradient staleness and often require storing historical weights, an impractical storage requirement in SL. Solutions like approximate weight updates \cite{yang2021pipemare} and Nesterov \cite{ajanthan2025nesterov} have been proposed; however, they have only been tested on small models (e.g., nanoGPT \cite{Karpathy2022}). In this work, we adopt synchronous training and the Nesterov‑based(without Weight Stashing)  method \cite{ajanthan2025nesterov} as baselines, and scale the evaluation to Qwen2.5‑1.5B while providing convergence guarantees.

\section{Methodology}
\subsection{Two Phase Optimization Scheme}

Without the loss of generality, we decompose the LLM into $\bNum$ blocks. Each of these blocks comprises several transformer blocks, and they are distributed between the client and server nodes. To preserve the residual architecture, we emphasize that individual transformer blocks cannot be split internally.
Thus, we can formulate the LLM training into following optimization problem:
\begin{equation} \label{eq:e2e_opt}
\begin{aligned}
    \min_{\lW_{1,...,\bNum}} \quad & \loss_{\text{task}}(\lAct_{\bNum}, Y) \quad 
    \text{s.t.} \ & \lAct_{l+1} = \lAct_l + \resF_l(\lAct_l; \lW_l), \quad \forall l \in \{1, \dots, \bNum-1\}
\end{aligned}
\end{equation}
The forward propagation is based on Equation~(\ref{eq:e2e_opt}). For backpropagation, we relax the strict equality constraints into a quadratic penalty method formulation. 
Specifically, we denote $\Theta=\{\lW_l\}_{l=1}^{\bNum - 1}$ and $H=\{\lAct_l\}_{l=1}^{\bNum}$. We make variables in $H$ free parameters and introduce a punishment term with a coefficient $\rho$, so the coupled problem is transformed into a globally relaxed objective:
\begin{equation} \label{eq:penalty_global}
    \min_{\Theta, H} \quad \loss(\Theta,H) =  \loss_{\text{task}}(\lAct_{\bNum}, Y) + \frac{\lRho}{2} \sum_{l=1}^{\bNum-1}  \left\| \lAct_{l+1} - \Big( \lAct_l + \resF_l(\lAct_l; \lW_l) \Big) \right\|^2
\end{equation}

For Equation (\ref{eq:penalty_global}), we employ a two-phase Gaussian-Seidel optimization scheme. 
We consider the iterative optimization in the time stamp $t$, and use the superscript $\cdot^{(t)}$ to indicate the value in the current time stamp. Without the loss of generality, we focus on the block $\lAct_{l+1} = \lAct_l + \resF_l(\lAct_l; \lW_l)$.

\textbf{Phase 1 (Cross Block).} We compute the gradient $\lGrad_{l + 1} = \nabla_{\lAct_{l + 1}} \loss(\Theta, H)$. For intermediate blocks, $\lGrad_{l + 1}$ relies on gradients backpropagated from the subsequent block and is calculated upon receiving them. We then consider a one-step gradient descent with step size $\lEta_{\lAct}$: $\lAct_{l + 1}^{(t+1)} = \lAct_{l + 1}^{(t)} - \lEta_{\lAct} \lGrad_{l + 1}^{(t)}$.


\textbf{Phase 2 (In Block).} We consider optimizing $\lW_l$: $\min_{\lW_l} \left\|\lAct_{l + 1}^{(t + 1)} - \Big(\lAct_l^{(t)} + F_l(\lAct_l^{(t)}; \lW_l)\Big)\right\|^2$. Here, we use $\lAct_{l + 1}^{(t + 1)}$, $\lAct_l^{(t)}$, because $\lW_l$ is updated after $\lAct_{l + 1}$ but before $\lAct_l$. Substituting $\lAct_{l + 1}^{(t+1)} = \lAct_{l + 1}^{(t)} - \lEta_{\lAct} \lGrad_{l + 1}^{(t)}$ in Phase 1 and forward propagation $\lAct_{l+1}^{(t)} = \lAct_l^{(t)} + \resF_l(\lAct_l^{(t)}; \lW_l^{(t)})$, we obtain:
\begin{equation} \label{eq:phase2_sub}
\begin{aligned}
    \lW_{l}^{(t+1)} &= \arg \min_{\lW_l} \left\| \Big( \lAct_{l+1}^{(t)} - \lEta_{\lAct} \lGrad_{l+1}^{(t)} \Big) - \Big( \lAct_l^{(t)} + \resF_l(\lAct_l^{(t)}; \lW_l) \Big) \right\|^2 \\
    & = \arg\min_{\lW_l} \left\| \resF_l(\lAct_l^{(t)}; \lW_l^{(t)}) - \resF_l(\lAct_l^{(t)}; \lW_l) - \lEta_{\lAct} \lGrad_{l+1}^{(t)} \right\|^2 \\
    & = \arg\min_{\lW_l}\frac{1}{2}\left\| \resF_l(\lAct_l^{(t)}; \lW_l^{(t)}) - \resF_l(\lAct_l^{(t)}; \lW_l) \right\|^2 + \lEta_{\lAct} \langle \resF_l(\lAct_l^{(t)}; \lW_l), \lGrad_{l+1}^{(t)} \rangle \\
\end{aligned}
\end{equation}
$\lAct_{l+1}^{(t + 1)}$ and $\lAct_{l+1}^{(t)}$ cancel out in Equation~(\ref{eq:phase2_sub}), so we only need $\lAct_l^{(t)}$, $\lW_l^{(t)}$ and $\lGrad_{l+1}^{(t)}$ to calculate $\lW_{l}^{(t+1)}$. This allows us to simplify Phase 1 to store merely the per‑block activations $\lAct_l^{(t)}$ and gradients $\lGrad_{l+1}^{(t)}$ for the subsequent parameter update without the explicit update of $\lAct_{l+1}^{(t+1)}$.


\subsection{CommuteProp Framework}
\label{sec:two_phase}

Under the two-phase Gauss‑Seidel optimization scheme discussed above, each block in SL must perform three distinct tasks: in Phase 1, (i) receiving outputs from the preceding block, computing and transmitting the activations $\lAct_{l+1}$ to the subsequent block during forward propagation; (ii) receiving gradients from the subsequent block, computing and transmitting the gradients $g_{l+1}$ to the preceding block during backpropagation; and in Phase 2, (iii) performing the in‑block update of parameters $\lW_{l}$.

Motivated by this analysis, we propose \textit{CommuteProp} (illustrated in Figure~\ref{fig:overall}), a framework that processes these tasks asynchronously across different data batches. Because the Phase 2 parameter updates rely on the gradients $g_{l + 1}$ computed during Phase 1, we establish two separate queues, $\gQ_1$ and $\gQ_2$, for each block. Phase 1 tasks are managed in $\gQ_1$, while the in-block parameter update tasks (Phase 2) are put in $\gQ_2$. 
These queues store only the essential contextual information required for each task, such as $\lAct_l$ and $g_{l + 1}$. To strictly control memory consumption, we cap the maximum length of each queue at $\tau_{\max}$. Furthermore, we exclude the parameter weights $\lW_l$ from the queued task information, as they typically possess a significantly larger memory footprint than either the activations or their gradients.

The exclusion of parameter weights $\lW_l$ from the queued task information introduces challenges when applying Equation~(\ref{eq:phase2_sub}) to update $\lW_l$. While we can use the delayed gradients $g_{l + 1}^{(t - \tau)}$ and activations $\lAct_l^{(t - \tau)}$ retrieved from the queue ($\tau \leq \tau_{\max}$), the absence of the corresponding historical parameters $\lW_l^{(t - \tau)}$ forces the update mechanism to rely on the current parameters $\lW_l^{(t)}$ instead.

For notational simplicity, we define the local loss as $\loss_{local}(\lW_l) = \langle \resF_l(\lAct_l^{(t-\tau)}; \lW_l), \lGrad_{l+1}^{(t-\tau)} \rangle$ and assume $F_l(\lAct_l^{(t - \tau)}; \cdot)$ is $L_0$-Lipschitz continuous: $\| F_l(\lAct_l^{(t - \tau)}; \lW_l) - F_l(\lAct_l^{(t - \tau)}; \lW_l^{(t)}) \| \le L_0 \| \lW_l - \lW_l^{(t)} \|$. We can then reformulate the update rule from Equation~(\ref{eq:phase2_sub}) as a proximal optimization problem: 

$$\lW_{l}^{(t+1)} = \arg\min_{\lW_l} \left\{ \lEta_{\lAct} \loss_{local}(\lW_l) + \frac{L_0^2}{2} \left\| \lW_l -\lW_l^{(t)} \right\|^2 \right\}$$
Furthermore, we linearize $\loss_{local}(\lW_l)$ at the current parameter $\lW_l^{(t)}$, giving the approximation:
\begin{equation}
\label{eq:update}
\begin{aligned}
\lW_{l}^{(t+1)} & \approx \arg\min_{\lW_l} \left\{ \lEta_{\lAct}  \left\langle \nabla_{\lW_l} \loss_{local}(\lW_l^{(t)}), \lW_l-\lW_l^{(t)} \right\rangle + \frac{L_0^2}{2} \left\| \lW_l -\lW_l^{(t)} \right\|^2 \right\} = \lW_l^{(t)} - \lEta \widetilde{\lGrad}_l^{(t-\tau)}
\end{aligned}
\end{equation}
where $\widetilde{\lGrad}_l^{(t-\tau)} = \nabla_{\lW_l} \loss_{local}(\lW_l^{(t)}) = \left[\nabla_{\lW_l} \resF_l(\lAct_l^{(t-\tau)}; \lW_l^{(t)})\right] \lGrad_{l+1}^{(t-\tau)}$ denotes the surrogate gradient, and $\eta = \lEta_{\lAct} / L_0^2$ is the effective learning rate which can be tuned in practice.

Notably, we use the current parameter $\lW^{(t)}$ instead of $\lW^{(t - \tau)}$ which is unsaved. A rigorous convergence analysis of this substitution is provided in Section \ref{sec:convergence}. Moreover, although the derivation assumes a first‑order update, it naturally extends to popular optimizers such as AdamW\cite{loshchilov2017decoupled} and Muon\cite{jordan2024muon} for LLM fine‑tuning. We formalize our two‑phase pipeline, CommuteProp, in Algorithm~\ref{alg:commuteprop}.

\begin{figure}[htbp]
    \centering 
    \begin{minipage}[t]{0.25\textwidth}
        \vspace{0cm} 
        \centering
        \includegraphics[height=0.25\textheight]{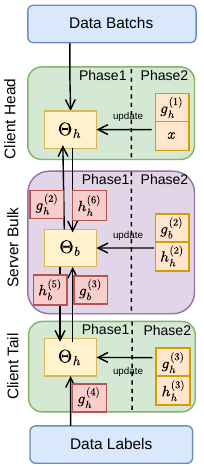}
        \captionsetup{font=scriptsize, justification=centering}
        \caption{Snapshot of Algorithm \ref{alg:commuteprop} on Usplit }
        \label{fig:alg}
    \end{minipage}%
    \hfill 
    \begin{minipage}[t]{0.7\textwidth}
        \vspace{-17pt}
        \begin{algorithm}[H]
            \footnotesize          
            \setlength{\itemsep}{0pt}   
            \setlength{\parsep}{0pt}
            \setlength{\topsep}{0pt}
            \setlength{\partopsep}{0pt}
            \captionsetup{font=footnotesize}
            \caption{CommuteProp Node: Event Ingestion }
            \label{alg:commuteprop}
            \begin{algorithmic}[1]
            \REQUIRE Phase 1 Task Queue $\mathcal{Q}_1$, Phase 2 Task Queue $\mathcal{Q}_2$, Accumulation bound $\tau$
            
            \WHILE{Training is active}
            
                \IF{\textbf{Received} Backward target $\lGrad_{l+1}^{(t)}$ from network}
                    \STATE $\mathcal{Q}_1.\text{Push}(\text{Task\_Backward}(\lGrad_{l}^{(t)}))$ \COMMENT{Backward pass $\Theta$ detached}
                
                \ELSIF{\textbf{Received} Forward request $\lAct_l^{(t)}$ from network}
                    \STATE $\mathcal{Q}_1.\text{Push}(\text{Task\_Forward} (\lAct_{l+1}^{(t)}))$ \COMMENT{Forward pass with $\Theta$ detached}
            
                \ENDIF
                
                \IF{Receive both $\lAct_{l+1}^{(t)}$ and $\lGrad_{l+1}^{(t)}$}
                    \STATE $\mathcal{Q}_2.\text{Push}(\text{Task\_Phase2\_Update}(\lAct_{l+1}^{(t)},\lGrad_{l}^{(t)})$ \COMMENT{Update $\lW_l$ based on (\ref{eq:update})}
                \ENDIF
                \STATE Task $\leftarrow$ \textbf{Dispatcher\_Schedule}($\mathcal{Q}_1,\mathcal{Q}_2$) 
                \COMMENT{Dispatcher Algorithm \ref{alg:dispatcher}}
                \IF{Task is not NULL}
                    \STATE \textbf{Execute}(Task) \COMMENT{Run the actual computation}
                \ENDIF
            \ENDWHILE
            \end{algorithmic}
        \end{algorithm}
    \end{minipage}
    \vspace{-0.15cm}
\end{figure}

Algorithm~\ref{alg:commuteprop} outlines the per-node execution logic, while Figure~\ref{fig:alg} illustrates a snapshot of the overall Usplit workflow at a specific time step. In this figure, the red squares denote the outputs of cross-block tasks during Phase~1. Conversely, the orange squares represent the inputs derived from these Phase~1 outputs, which subsequently trigger the Phase~2 tasks. Notably, the information communicated between blocks consists strictly of intermediate results ($\lAct_l$ and $\lGrad_{l+1}$) rather than actual model parameters. Because of this, CommuteProp is highly compatible with existing privacy-preserving techniques~\cite{liu2025dualguard}. This compatibility allows the intermediate data to be safely offloaded to a memory-abundant server, effectively shifting the heavy memory burden away from the client. Finally, given the asynchronous nature of this execution model, efficiently scheduling these distinct tasks to maximize overall pipeline throughput presents a critical challenge, which we address in the subsequent section.

\subsection{Dispatcher Schedule and Time Efficiency}

In the CommuteProp framework (Algorithm~\ref{alg:commuteprop}), client and server nodes schedule forward activations, backward activation gradients (Phase~1), and weight updates (Phase~2). 
The scheduling strategy is crucial for computational efficiency, but finding the optimal scheduling strategy is a computationally prohibitive mixed‑integer programming (MIP) problem. 
In this context, based on the Critical Path Method (CPM)~\citep{kelley1959critical}, we introduce a lightweight CPM Heuristic Dispatcher as shown in Algorithm~\ref{alg:dispatcher}. Specifically, Phase~1, which are stored in $\gQ_1$, naturally form the critical path because they produce the activations and gradients required for communication. Any delay in them directly stalls downstream stages. Phase~2 weight updates stored in $\gQ_2$, by contrast, are strictly local and lie off the critical path. Thus, in Algorithm \ref{alg:dispatcher}, we assign top priority to Phase 1. To avoid excessive staleness, we also set a maximum latency $\tau_{\max}$ for Phase 2.

\begin{figure}[htbp]
  \vspace{-0.2cm}
  \centering
  \begin{minipage}[t]{0.65\textwidth}
    \vspace{0pt} 
    \centering
    
    \begin{subfigure}[b]{0.49\linewidth}
    \vspace{0.1cm}
      \centering
      \includegraphics[width=\linewidth]{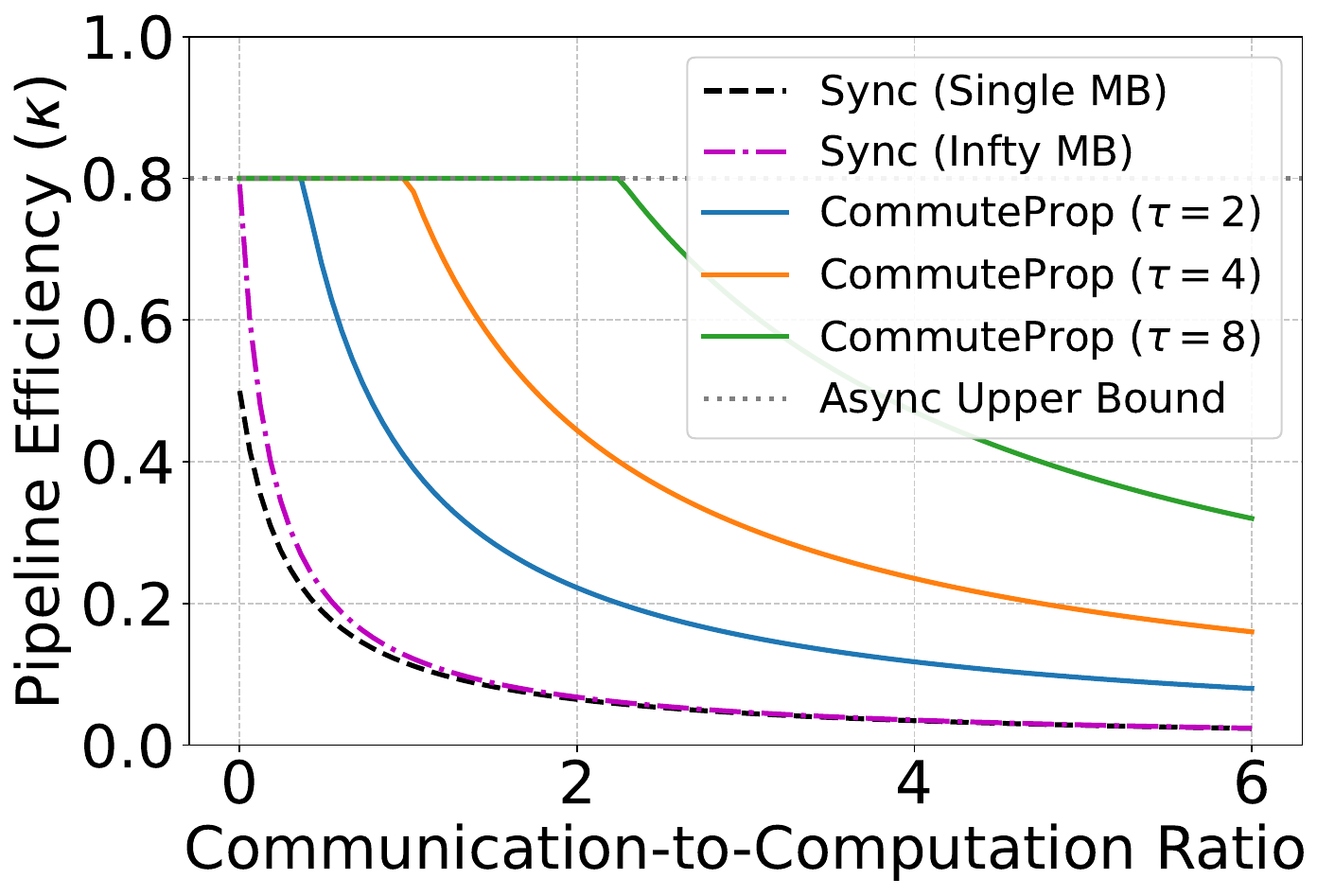} 
      \caption{}
      \label{fig:eff_comm}
    \end{subfigure}
    \hfill
    \begin{subfigure}[b]{0.49\linewidth}
      \centering
      \includegraphics[width=\linewidth]{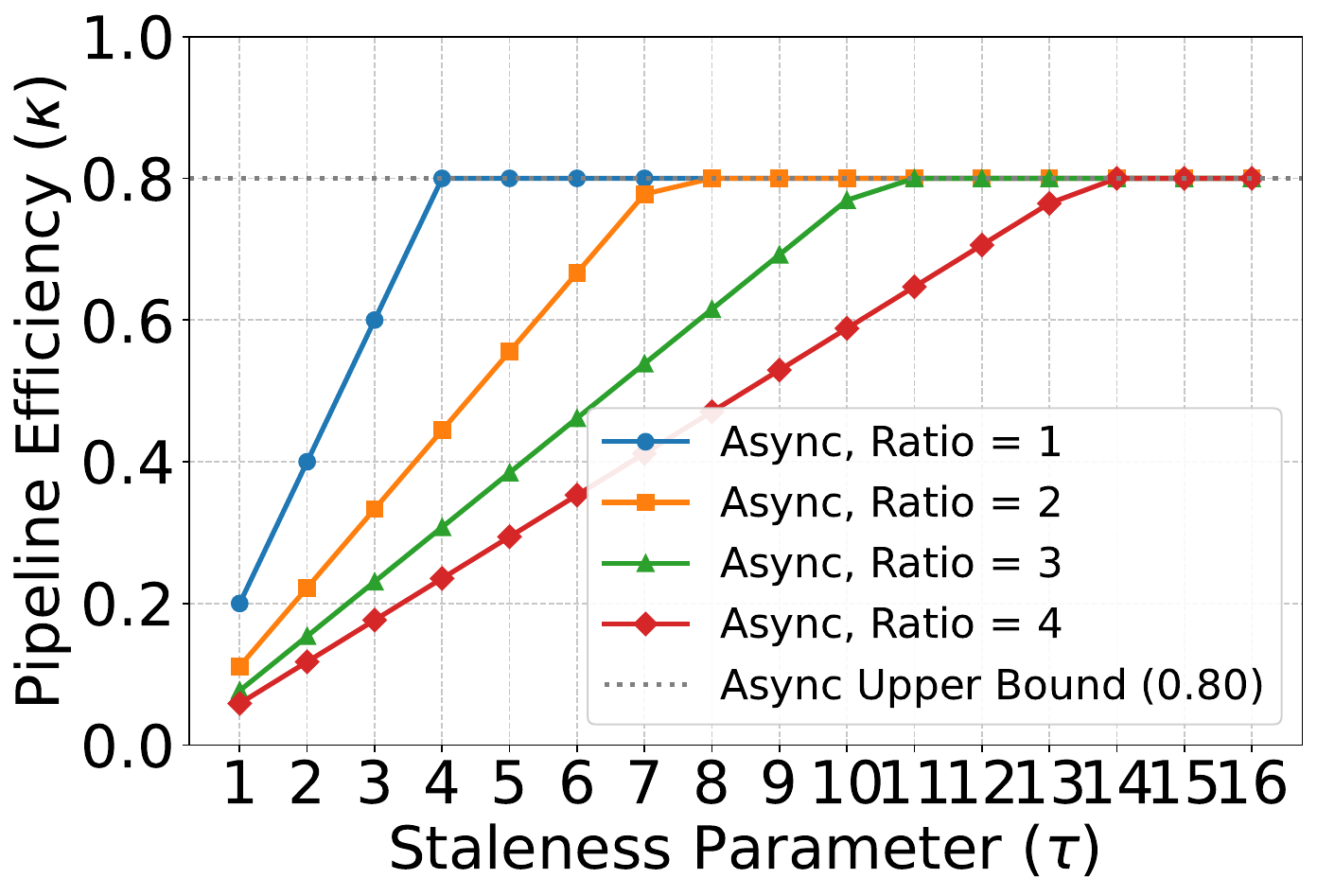}
      \caption{}
      \label{fig:eff_tau}
    \end{subfigure}
    \captionsetup{font=footnotesize}
    \caption{Pipeline efficiency comparison across different system parameters. Figure \ref{fig:eff_comm} demonstrate the relationship between pipeline efficiency and communication computation ratio. Figure \ref{fig:eff_tau} demonstrate the relationship between $\tau$ and pipeline efficiency } 
    \label{fig:efficiency_plots_horizontal}
  \end{minipage}%
  \hfill
  \begin{minipage}[t]{0.33\textwidth}
    \vspace{-1em} 
    \begin{algorithm}[H] 
        
        \footnotesize          
        \setlength{\itemsep}{0pt}   
        \setlength{\parsep}{0pt}
        \captionsetup{font=footnotesize}
        \caption{CPM Heuristic Dispatcher}
        \label{alg:dispatcher}
        \begin{algorithmic}[1]
        \REQUIRE Phase 1 Queue $\mathcal{Q}_1$, Phase 2 Queue $\mathcal{Q}_2$
        \WHILE{Training is active}
            \IF{$|\mathcal{Q}_2| \ge \tau_{\max}$}
                \STATE $\mathcal{Q}_2.\text{Pop()}$
            \ELSIF{$|\mathcal{Q}_1| > 0$}
                \STATE $\mathcal{Q}_1.\text{Pop()}$
            \ELSIF{$|\mathcal{Q}_2| > 0$}
                \STATE $\mathcal{Q}_2.\text{Pop()}$
            \ENDIF
        \ENDWHILE
        \end{algorithmic}
    \end{algorithm}
  \end{minipage}
  \vspace{-0.2cm}
\end{figure}

In CommuteProp, $t_{l,1}$ and $t_{l,2}$ represent the respective computation durations of Phase 1 and Phase 2 for node $l$, while $t_{c}$ denotes the one-way communication time. For classical pipeline synchronous algorithm with micro-batch $m$,  denote the total computation time for node $l$ as $t_{l,s}$. Denote $\kappa_a$ and $\kappa_s$ as  pipeline efficiency for asynchronous CommuteProp and synchronous algorithms, we have:
\vspace{-0.5em}
\begin{equation}
\label{equ:efficiency}
\resizebox{\linewidth}{!}{$
    \kappa_{a} = \frac{\sum_{l=1}^L (t_{l,1} + t_{l,2})}{ L \max\left\{\max\limits_{l} \{t_{l,1} + t_{l,2}\}, \frac{2t_c+\sum_{l=1}^L t_{l,1} }{ \tau_{\max}} \right\} }\; \kappa_{s}(M) = \frac{M \sum_{l=1}^{L} t_{l,s}}{L \left[ \sum_{l=1}^{L} t_{l,s} + 2t_c + (M-1)\left(2t_c + \max_{l} t_{l,s}\right) \right]}$}
\vspace{-0.5em}
\end{equation}

In CommuteProp, the denominator’s bottleneck term comes either from the heaviest node $\max\limits_{l} \{t_{l,1} + t_{l,2}\}$ or from the overall communication time $\left(2t_c+\sum_{l=1}^L t_{l,1} \right)/{ \tau_{\max}}$. The synchronous denominator, in contrast, incurs a per‑node serialisation of $2t_c+\max_l t_{l,s}$ and, additionally, a pipeline‑drain cost of $\sum_l t_{l,s}$, which CommuteProp avoids entirely.

Specialising to the Usplit setting ($L=2$), we measure the communication impact via the ratios $r_{\text{async}} = 2t_c / (t_{1,1} + t_{1,2} + t_{2,1} + t_{2,2})$, 
while for the synchronous counterpart it is $r_{\text{sync}} = 2t_c / (t_{1,s} + t_{2,s})$. Because the split design introduces an extra forward computation, the synchronous pipeline’s computational load is heavier; In Figure Figure~\ref{fig:eff_comm}, we assume in $r_{\text{sync}} \approx 0.75 \, r_{\text{async}}$ and server is 1.67 faster than client.  Figure~\ref{fig:eff_comm} shows that CommuteProp yields a significant improvement in pipeline efficiency, 
while Figure~\ref{fig:eff_tau} reveals that systems with a higher communication-to-computation ratio 
require a correspondingly larger $ \tau_{\max}$ to achieve optimal pipeline efficiency. In the next section, we theoretically analyze the relationship between the latency threshold $ \tau_{\max}$ and the resulting utility.

\section{Convergence Analysis}
\label{sec:convergence}
\subsection{Preliminary}
\label{sec:preliminary}
We begin by introducing the following assumptions:

\begin{assumption}[Hierarchical Lipschitz Smoothness]
\label{assum:cross_hessian}
We assume the loss function $\loss$ is twice-differentiable and globally $L$-Lipschitz smooth, such that its Hessian is bounded by $L$, i.e., $\| \nabla^2 \loss(\lW) \| \le L$. Furthermore, we assume a hierarchical smoothness structure: for any individual block $l$, the block-wise Hessian norm is bounded by $L_l$, and for any pair of distinct blocks $l$ and $j$ ($l \neq j$), the cross-block Hessian norm is bounded by $H_{l,j} \ge 0$. Formally, we have:
\begin{equation}
    \left\| \nabla^2_{\lW_l \lW_l} \loss(\lW) \right\| \le L_l, \quad \text{and} \quad \left\| \nabla^2_{\lW_l \lW_j} \loss(\lW) \right\| \le H_{l,j}
\end{equation}
\end{assumption}

\begin{assumption}[Bounded Gradients]
\label{assum:bounded_gradient}
We assume that the objective function has bounded gradients. That is, there exists a constant $G > 0$ such that for any parameter state $\lW$, the global gradient norm is bounded by $\| \nabla_{\lW} \loss(\lW) \| \le G.$
\end{assumption}

Assumption \ref{assum:cross_hessian} characterizes the hierarchical smoothness properties of the objective function. Crucially, recent empirical and theoretical studies \cite{zhang2024transformers,dong2025towards} on Transformer architectures  demonstrate that the cross-block Hessian norms $H_{l,j}$ are typically orders of magnitude smaller than the block-wise Hessian norms $L_l$. We empirically verify this property on the evaluated split architectures in Section~\ref{sec:hessian}.

Furthermore, under Assumption \ref{assum:bounded_gradient}, if we employ a standard first-order optimization method (such as Stochastic Gradient Descent) with a learning rate $\lEta$, the parameter update rule $\lW^{(t)} = \lW^{(t-1)} - \lEta \nabla \loss(\lW^{(t-1)})$ naturally implies that the step size for any block $j \in \{1, \dots, \bNum\}$ is strictly bounded per iteration: $\| \lW_j^{(t)} - \lW_j^{(t-1)} \| \le \lEta G$. Further convergence analysis relies on this bounded update property in our subsequent convergence analysis.

\subsection{Asynchronous Error Analysis}
Although CommuteProp updates blocks asynchronously, we analyze its convergence 
by a virtual serialized sequence in which blocks $1,\dots,L$ are updated one by one, 
each using the same delayed gradient as in the real execution. 
Let $\Theta^{(t)}$ be the parameters after $t$ steps in this virtual sequence 
($\Theta^{(0)}$ is the initial point). Then
\begin{equation}
\vspace{-0.5em}
    \loss(\Theta^{(L)}) - \loss(\Theta^{(0)}) = \sum_{t=0}^{L-1} \bigl( \loss(\Theta^{(t+1)}) - \loss(\Theta^{(t)}) \bigr)
\end{equation}
In step $t$, only one block is updated; the other blocks remain fixed. 
Writing the loss as a function of $\lW_l$ alone (still denoted $\loss(\lW_l)$) and using $L_l$-smoothness yields
\begin{equation}
    \loss(\lW_{l}^{(t+1)}) \le \loss(\lW_{l}^{(t)}) + \langle \nabla_{\lW_l}\loss(\lW_{l}^{(t)}),\lW^{(t+1)}_l-\lW^{(t)}_l\rangle + \frac{L_l}{2}\|\lW^{(t+1)}_l-\lW^{(t)}_l\|^2 
\end{equation}
Substituting the update \eqref{eq:update} gives
\begin{equation} \label{eq:descent}
\vspace{-0.5em}
     \loss(\lW_{l}^{(t+1)}) \le \loss(\lW_{l}^{(t)}) - \frac{\eta}{2}\|\nabla_{\lW}\loss(\lW_l^{(t)})\|^2 - \left( \frac{\eta}{2} - \frac{L_l\eta^2}{2} \right)\|\widetilde{\lGrad}^{(t- \tau_{\max})}\|^2 + \frac{\eta}{2}\|\nabla_{\lW}\loss(\lW_l^{(t)})-\widetilde{\lGrad_l}^{(t- \tau_{\max})}\|^2 
\end{equation}

Equation (\ref{eq:descent}) indicates that the per-step descent is bottlenecked by asynchronous error $\| \nabla_{\lW}\loss(\lW_l^{(t)})-\widetilde{\lGrad}^{(t-  \tau)}\|^2$. Thus, we have the following Lemma:

\begin{lemma}[Asynchronous Error Bound]
\label{lemma:async_error}
Let $\lW^{(t)} = [\lW_1^{(t)}, \dots, \lW_\bNum^{(t)}]$ be the state of the model parameters at step $t$. For a target layer $l$ under a staleness delay of $  \tau\le\tau_{\max}$ steps, the asynchronous proxy gradient $\widetilde{\lGrad}_l^{(t- \tau)}$ satisfies the following error bound:
\begin{equation}
    \left\| \nabla_{\lW_l}\loss(\lW^{(t)}) - \widetilde{\lGrad}_l^{(t-  \tau)} \right\|^2 \le \left( \lEta   \tau_{\max} G \sum_{j \neq l} H_{l,j} \right)^2
\end{equation}
\vspace{-0.5em}
\end{lemma}

Lemma \ref{lemma:async_error} reveals that the asynchronous error inherently scales with learning rate, the staleness delay $  \tau$ and the cross-block Hessian norms $H_{l,j}$. As highlighted in Section \ref{sec:preliminary}, the cross-block curvature $H_{l,j}$ is exceptionally small for $l \neq j$. Consequently, the summation term $\sum_{j \neq l} H_{l,j}$ remains relatively small. Building upon this lemma, we establish the overall convergence guarantee:

\begin{theorem} \label{thm:convergence}
Let $\loss^*$ be the global lower bound of the task loss function. Under Assumptions \ref{assum:cross_hessian} and \ref{assum:bounded_gradient}, if the learning rate satisfies $\lEta \le \frac{1}{L_l}$ (where $L_l$ is the Lipschitz smoothness constant of block $l$), the CommuteProp algorithm satisfies the following ergodic convergence bound after $T$ iterations:
\begin{equation}
\frac{1}{T} \sum_{t=0}^{T-1} \| \nabla_{\lW_l} \loss(\lW^{(t)}) \|^2 \le \frac{2[\loss(\lW_l^{(0)}) - \loss^*]}{\lEta T} + \lEta^2   \tau_{\max}^2 G^2 \left( \sum_{j \neq l} H_{l,j} \right)^2
\vspace{-0.5em}
\end{equation}
\end{theorem}

Theorem \ref{thm:convergence} demonstrates that CommuteProp converges to an $\varepsilon$-approximate stationary point. The approximation error, characterized by $\varepsilon = \mathcal{O}\left(\lEta^2   \tau_{\max}^2 G^2 (\sum_{j \neq l} H_{l,j})^2\right)$, acts as a constant noise floor. In practice, as the model approaches the stationary point, the true gradient norm shrinks, meaning the effective local gradient bound $G^{(t)}$ naturally diminishes. In the next section, we will discuss how to further reduce the asynchronous error for preconditioned optimizers.

\subsection{Asynchronous Error Reduction Method}

The theoretical results in Lemma \ref{lemma:async_error} and Theorem \ref{thm:convergence} establish that the convergence of decoupled training is inherently bottlenecked by the accumulation of staleness-induced noise. To further dissect this mechanism, let $\mathbf{M}_j^{(t)}$ denote the adaptive scaling operator of the optimizer at step $t$ (e.g., $\mathbf{M}_j^{(t)}=\mathbf{I}$ for SGD, or $\mathbf{M}_j^{(t)}=\text{diag}(1/\sqrt{\mathbf{v}_t+\varepsilon})$ for AdamW). By unpacking the explicit update step, the asynchronous error bound can be reformulated via the local preconditioned updates:

\vspace{-0.5em}
\begin{equation}
\label{equ:general_error}
\left\| \nabla_{\lW_l}\loss(\lW^{(t)}) - \widetilde{\lGrad}_l^{(t-  \tau)} \right\|^2 \le \left( \lEta   \tau_{\max} \sum_{j \neq l} \left\| \mathbf{H}_{l,j} \mathbf{M}_j^{(t)} \widetilde{\lGrad}_j^{(t-  \tau)} \right\| \right)^2
\vspace{-0.5em}
\end{equation}

\begin{wrapfigure}{r}{0.6\textwidth}
  \centering
  \vspace{-2pt} 
  \begin{minipage}{\linewidth}
    \hrule height .8pt
    \vspace{3pt}
    \footnotesize
    \renewcommand{\baselinestretch}{0.85}
    \captionsetup{type=algorithm, font={footnotesize,bf}, skip=5pt, justification=raggedright, singlelinecheck=false}
    \caption{Kronecker-Block NS Preconditioner}
    \label{alg:ns_precond}
    \hrule height .5pt
    
    \begin{algorithmic}[1]
      \setlength{\itemsep}{1pt}
      \setlength{\parsep}{0pt}
      
      \REQUIRE Gradient matrix $U \in \mathbb{R}^{M \times N}$, block size $b$.
      
      \STATE Partition $U$ into an $R \times C$ grid of blocks $\{U_{rc}\}_{r,c=1}^{R,C}$.
      
      \STATE Maintain row covariances $\{V_r\}_{r=1}^R$ and column covariances $\{V_c\}_{c=1}^C$; row preconditioners $\{P_r\}_{r=1}^R$ and column preconditioners $\{P_c\}_{c=1}^C$. 
      
      \STATE $V_r \gets \beta V_r + (1-\beta)\sum_c U_{rc} U_{rc}^\top$
      
      \STATE $V_c \gets \beta V_c + (1-\beta)\sum_r U_{rc}^\top U_{rc}$ 
      \hfill $\triangleright$ \textbf{Low-freq covariance update}
      
      \STATE $P_r \gets V_r^{-1/4}$ and $P_c \gets V_c^{-1/4}$ \\
      \quad via Newton-Schulz iteration. \hfill $\triangleright$ \textbf{Preconditioner refresh}
    
      \STATE \textbf{return} $\tilde{U}$ with blocks $\tilde{U}_{rc} = P_r \, U_{rc} \, P_c$.
    \end{algorithmic}
    \vspace{2pt}
    \hrule height .8pt
  \end{minipage}
  \vspace{-10pt}
\end{wrapfigure}

Equation (\ref{equ:general_error}) directly indicates that the asynchronous error is driven by the interaction between the cross-block Hessian $\mathbf{H}_{l,j}$ and the scaled update direction $\mathbf{M}_j \widetilde{\lGrad}_j^{(t-  \tau)}$. In standard AdamW, $\mathbf{M}_j$ serves as a purely diagonal approximation of the local inverse covariance. Because it inherently fails to decorrelate the off-diagonal components of the local curvature ($\mathbf{H}_{j,j}$), the update direction $\mathbf{M}_j \widetilde{\lGrad}_j$ remains heavily biased toward the local high-curvature directions. When transmitted across blocks, these residual misaligned components can be amplified by the cross-block operator $\mathbf{H}_{l,j}$, exacerbating the asynchronous error accumulation.

To mitigate this, inspired by the Muon optimizer in first-order orthogonalization, we propose to explicitly \textit{whiten} the gradient space prior to applying the diagonal operator $\mathbf{M}_j$. By introducing a preconditioner $\mathcal{P}_j \approx \mathbf{H}_{j,j}^{-1/2}$, we transform the unbounded inter-block covariance into a bounded correlation operator, thereby neutralizing the asynchronous error. To achieve this inverse square root approximation efficiently without heavy eigendecomposition, we introduce the Kronecker-Block Newton-Schulz (NS) Preconditioner. For 2D parameter tensors (e.g., weight matrices), let $\mathbf{U} \in \mathbb{R}^{M \times N}$ denote the reshaped gradient matrix of the block update $\widetilde{\lGrad}_l$, as detailed in Algorithm \ref{alg:ns_precond}.


To ensure efficiency, Algorithm \ref{alg:ns_precond} employs a Kronecker-factored structure (inspired by SOAP\cite{vyas2024soap}) to tractably approximate high-dimensional second-moment statistics without materializing the full Hessian. Furthermore, since the empirical Fisher information evolves slowly during training \cite{zhao2024galore}, we implement lazy updates for both covariance aggregation and preconditioner generation. This spatio-temporal design significantly amortizes computational overhead while strictly preserving the spatial decorrelation crucial for robust asynchronous decoupling.

\section{Experiment}


\textbf{Experiment Settings.} We use Qwen2.5-1.5B-Instruct as the backbone and partition it with the Usplit architecture. To compare utility after full training without the high wall-clock cost of running every setting on real distributed systems, we use a queue-based pipeline simulator following prior split-learning works~\cite{liu2025dualguard,thapa2022splitfed,chen2024unveiling}. Forward activations are buffered up to depth $\tau_{\max}$ before each gradient update. We consider two settings: \textit{Full Training}, where asynchronous methods complete all epochs, and \textit{Equal Time}, where all methods use the same wall-clock budget (``Time constraint'' in Table~\ref{tab:utl}). 
The simulator assumes an ideal bubble-free schedule for the synchronous baseline, so the reported speedups are conservative. \textbf{In section \ref{tab:real_deploy}, we further supplement the simulation with real client-server experiments under controlled bandwidth and latency, together with PipeDream\citep{narayanan2021memory} comparisons.} Appendix~\ref{sec:appendix_impl} provides these results, simulation details, datasets, and hyperparameters.

\subsection{Model Utility}
\label{sec:utility}

Table \ref{tab:utl} reports utility across the full‑training and Iso‑time settings.We use Nesterov Async~\citep{ajanthan2025nesterov}, the most recent asynchronous method, as our asynchronous algorithm baseline, as it does not require additional weight stashing similar to our algorithm. AdamW and Muon in Table \ref{tab:utl} reports the synchronous result. Rows without “Time constraint” (e.g., AdamW, CommuteProp, Nesterov Async) correspond to full training; those marked “Time constraint” are obtained under the Iso‑time budget. Result on other base model are defferd in Appendix \ref{sec:addition_basemodel}

\begin{table}[htbp]
\vspace{-0.8em}
  \centering
  \footnotesize 
  \captionsetup{font=footnotesize, labelfont=bf} 
  \captionsetup{font=footnotesize}
  \caption{Utility comparison under full-training and Equal time budgets}
  \vspace{-1em}
  \label{tab:utl}
  \setlength{\tabcolsep}{0pt} 
  \renewcommand{\arraystretch}{0.85} 
  \begin{tabular*}{\textwidth}{@{\extracolsep{\fill}} l cccccc}
    \toprule
    & \multicolumn{2}{c}{\textbf{Classification (Acc)}} & \multicolumn{4}{c}{\textbf{Generation}} \\
    \cmidrule(lr){2-3} \cmidrule(lr){4-7}
    \textbf{Algorithm} & Emotional & AG-News & \multicolumn{2}{c}{SAMSum} & SciERC & Spider \\
    \cmidrule(lr){4-5} \cmidrule(lr){6-6} \cmidrule(lr){7-7}
    & Accuracy & Accuracy & ROUGE & METEOR & F1-Score & EM \\
    \midrule
    Zero-shot                 & 43.33 & 63.00 & 23.82 & 37.82 & 4.30  & 36.00 \\
    Nesterov Async            & 76.45 & 90.56 & 35.53 & 38.13 & 7.73  & 71.00 \\
    \midrule
    AdamW                     & 91.15 & 94.64 & 36.61 & 39.72 & 72.89 & 81.40 \\
    AdamW (Time)              & 88.45 & 94.15 & 35.42 & \textbf{40.01} & 63.10 & 80.08 \\
    CommuteProp               & \textbf{90.85} & 94.57 & 36.49 & 39.45 & 67.21 & \textbf{81.60} \\
    CommuteProp (NS)          & 90.10 & \textbf{94.61} & \textbf{36.56} & 39.41 & \textbf{72.09} & 81.00 \\
    \midrule
    Muon                      & 93.10 & 92.97 & 36.27 & 39.29 & 72.12 & 76.20 \\
    Muon (Time)               & 92.00 & \textbf{92.05} & 35.47 & \textbf{39.69} & 60.68 & 72.60 \\
    CommuteProp (Muon)        & \textbf{92.60} & 92.00 & \textbf{36.31} & 39.18 & \textbf{65.15} & \textbf{73.02} \\
    \bottomrule
  \end{tabular*}
  \vspace{-1.0em}
\end{table}

The result indicates that CommuteProp maintains accuracy comparable to synchronous AdamW across all tasks,while the NS preconditioner further narrows these already small gaps. \textbf{Since a large epoch budget is used to evaluate algorithm's utility and stability}, all methods approach nearly convergence under the Iso‑time setting, leaving only modest accuracy gaps; the throughput advantage of CommuteProp is more clearly seen in the loss curves (Fig.,\ref{fig:ablation}).

\subsection{Ablation Study}
\label{sec:ablation}
In this section, we conduct a sensitivity analysis to examine the impact of the staleness parameter $ \tau_{\max}$ and the model split architecture(denoted as $[head, bulk, tail]$, representing the number of layers placed in the first client part, the server part, and the final client part) on the performance of CommuteProp. We evaluate the performance \textbf{across 36 configurations} on the SciERC dataset.
\begin{figure}[h]
    \centering
    \includegraphics[width=\linewidth, height=0.65\linewidth]{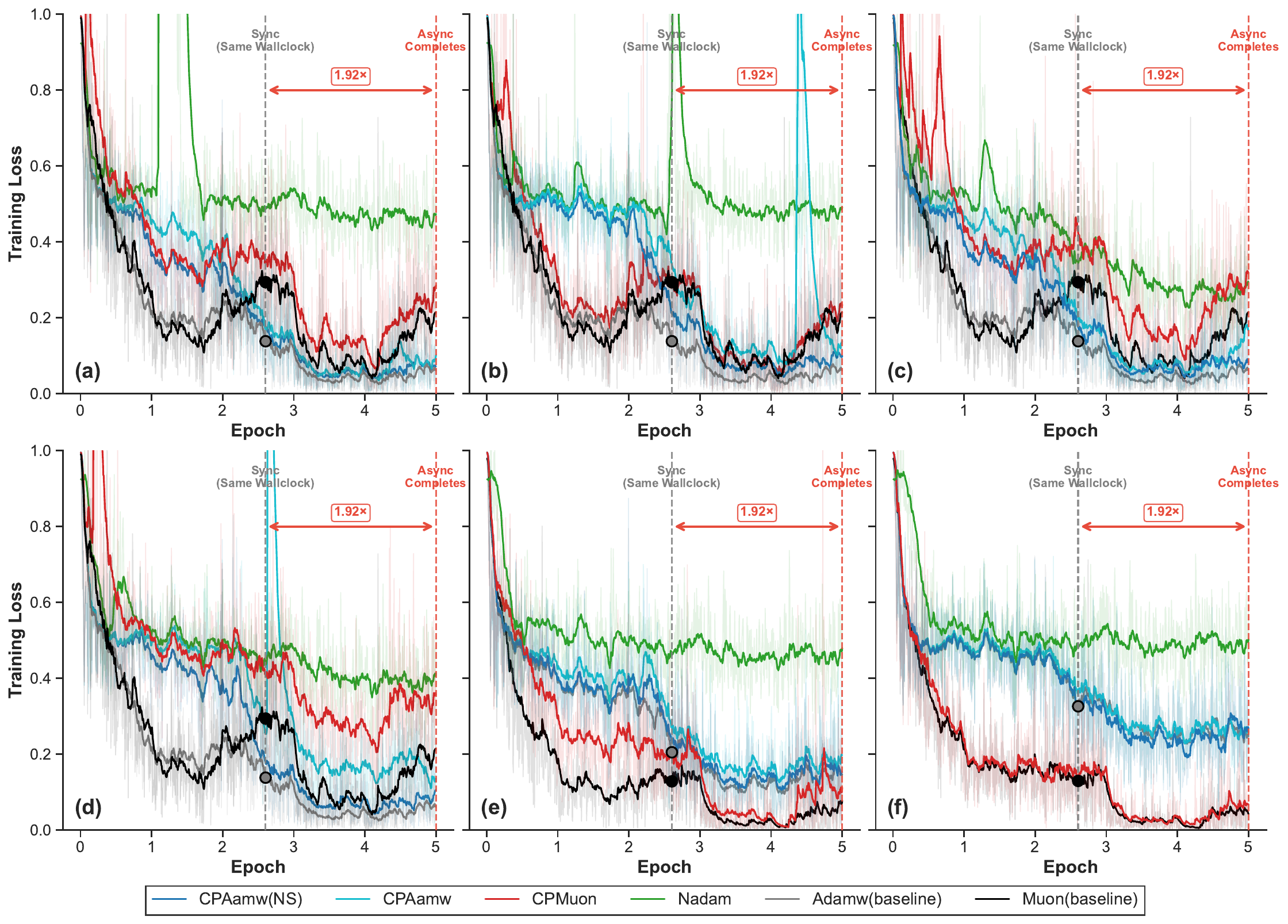}
  \captionsetup{font=footnotesize}
   \caption{
        Ablation study on the effects of $ \tau_{\max}$, $\lEta$, and split architecture(CP, short for Commutepro).
        (\textbf{a}--\textbf{c}) $\lEta = 5 \times 10^{-5}$, architecture $[3,22,3]$, with $ \tau_{\max} = 16$, $8$, and $4$, respectively.
        (\textbf{d}) $\lEta = 5 \times 10^{-5}$ and $ \tau_{\max} = 16$,  architecture $[1,26,1]$.
        (\textbf{e}--\textbf{f}) $ \tau_{\max} = 16$, architecture $[3,22,3]$, with $\lEta = 2 \times 10^{-5}$ and $1 \times 10^{-5}$, respectively.
    }
    \label{fig:ablation}
    \vspace{-0.5cm}
\end{figure}

Figure \ref{fig:ablation} verified our asynchronous error analysis(Lemma \ref{lemma:async_error}). Specifically, smaller $ \tau_{\max}$ and $\lEta$ lead to a smaller asynchronous error. Moreover, the NS preconditioner markedly reduces the asynchronous error, as its curve lies closer to the synchronous baseline. In addition, our results indicate that the proposed algorithm consistently outperforms the Nadam algorithm.

\subsection{Privacy Enhancement Integration}
\label{sec:pp}

\begin{wraptable}{r}{0.59\textwidth} 
  
  \centering
  \footnotesize
  \setlength{\tabcolsep}{1.4pt}
  \vspace{-2em} 
  \captionsetup{font=footnotesize, labelfont=bf}
  \caption{Model Utility and Attack Performance.}
  \label{tab:attack_performance_wrap}
  \begin{tabular}{lcccccc}
    \toprule
    & \multicolumn{2}{c}{\textbf{Performance}} & \multicolumn{4}{c}{\textbf{Attack Method}} \\
    \cmidrule(lr){2-3} \cmidrule(lr){4-7}
    & Rouge & Meteros & SIP & TAG & LAMP & CPTA \\
    \midrule
    w/o defense      & 86.23 & 88.48 & 86.87 & 65.70 & 63.13 & 17.75 \\
    Lora DualGuard   & 85.12 & 86.92 & 5.64  & 5.06  & 5.81  & 3.54  \\
    Lora CommuteProp & 85.18 & 87.13 & 5.64  & 4.34  & 4.57  & 3.58  \\
    \bottomrule
  \end{tabular}
  \vspace{-10pt} 
\end{wraptable}

Recent work \cite{gu2025vflair, chen2024unveiling}  shows that SL is vulnerable to gradient matching and Smd data inversion attacks. DualGuard~\citep{liu2025dualguard} was proposed as a defense. Thus, we integrate CommuteProp into DualGuard. Table~\ref{tab:attack_performance} reports the resulting utility and resilience against various attacks (detailed in Appendix \ref{sec:more_result}). The result indicates that CommuteProp achieves stronger privacy protection than synchronous DualGuard with negligible utility degradation. We attribute this improvement to the implicit regularization induced by stale activations. More experiment result are deferred to Appendix~\ref{sec:more_result}


\section{Conclusion and Future Work}
We proposed CommuteProp, an asynchronous split‑learning algorithm that decouples training into cross‑block and in‑block phases to overlap communication with computation. By exploiting the inherently small cross‑block Hessian of LLMs, we theoretically proved that our algorithm is resistant to asynchronous errors, and we accordingly introduced the NS preconditioner to further mitigate residual noise. Experiments show that CommuteProp achieves substantial throughput gains, maintains accuracy, and can integrate with or even enhance existing privacy‑preserving methods.

Notably, our theoretical analysis is not confined to Transformers and can be applied to feedforward networks as well. With an appropriate partition, the cross‑block Hessian norm could remain bounded, thus carrying over the convergence guarantees. However, identifying an optimal splitting strategy for feedforward architectures remains challenging, because they lack the shortcut mechanism that naturally induces small inter‑block interaction. Naive splitting may inflate the cross‑block Hessian and break our assumptions. We consider the design of principled splitting rules for general architectures as a meaningful direction for future work.



\bibliographystyle{plainnat}
\bibliography{ref}

@article{gupta2018distributed,
  title={Distributed learning of deep neural network over multiple agents},
  author={Gupta, Otkrist and Raskar, Ramesh},
  journal={Journal of Network and Computer Applications},
  volume={116},
  pages={1--8},
  year={2018},
  publisher={Elsevier}
}

@article{vepakomma2018split,
  title={Split learning for health: Distributed deep learning without sharing raw patient data},
  author={Vepakomma, Praneeth and Gupta, Otkrist and Swedish, Tristan and Raskar, Ramesh},
  journal={arXiv preprint arXiv:1812.00564},
  year={2018}
}

@inproceedings{liu2025dualguard,
  title={DualGuard: A Parameter Space Transformation Approach for Bidirectional Defense in Split-Based LLM Fine-Tuning},
  author={Liu, Zihan and Wang, Yizhen and Wang, Rui and Wu, Sai},
  booktitle={Proceedings of the 63rd Annual Meeting of the Association for Computational Linguistics (Volume 1: Long Papers)},
  pages={17065--17080},
  year={2025}
}

@article{gu2026differentially,
  title={Differentially Private and Communication Efficient Large Language Model Split Inference via Stochastic Quantization and Soft Prompt},
  author={Gu, Yujie and Jin, Richeng and Ji, Xiaoyu and Jin, Yier and Xu, Wenyuan},
  journal={arXiv preprint arXiv:2602.11513},
  year={2026}
}

@inproceedings{lyu2023optimal,
  title={Optimal resource allocation for u-shaped parallel split learning},
  author={Lyu, Song and Lin, Zheng and Qu, Guanqiao and Chen, Xianhao and Huang, Xiaoxia and Li, Pan},
  booktitle={2023 IEEE Globecom Workshops (GC Wkshps)},
  pages={197--202},
  year={2023},
  organization={IEEE}
}

@article{harlap2018pipedream,
  title={Pipedream: Fast and efficient pipeline parallel dnn training},
  author={Harlap, Aaron and Narayanan, Deepak and Phanishayee, Amar and Seshadri, Vivek and Devanur, Nikhil and Ganger, Greg and Gibbons, Phil},
  journal={arXiv preprint arXiv:1806.03377},
  year={2018}
}

@article{qi2023zero,
  title={Zero bubble pipeline parallelism},
  author={Qi, Penghui and Wan, Xinyi and Huang, Guangxing and Lin, Min},
  journal={arXiv preprint arXiv:2401.10241},
  year={2023}
}

@article{liu2024deepseek,
  title={Deepseek-v3 technical report},
  author={Liu, Aixin and Feng, Bei and Xue, Bing and Wang, Bingxuan and Wu, Bochao and Lu, Chengda and Zhao, Chenggang and Deng, Chengqi and Zhang, Chenyu and Ruan, Chong and others},
  journal={arXiv preprint arXiv:2412.19437},
  year={2024}
}

@inproceedings{gu2025vflair,
  title={Vflair-llm: A comprehensive framework and benchmark for split learning of llms},
  author={Gu, Zixuan and Fan, Qiufeng and Sun, Long and Liu, Yang and Ye, Xiaojun},
  booktitle={Proceedings of the 31st ACM SIGKDD Conference on Knowledge Discovery and Data Mining V. 2},
  pages={5470--5481},
  year={2025}
}

@article{kairouz2021advances,
  title={Advances and open problems in federated learning},
  author={Kairouz, Peter and McMahan, H Brendan},
  journal={Foundations and trends in machine learning},
  volume={14},
  number={1-2},
  pages={1--210},
  year={2021},
  publisher={Emerald Publishing Limited}
}

@article{vaswani2017attention,
  title={Attention is all you need},
  author={Vaswani, Ashish and Shazeer, Noam and Parmar, Niki and Uszkoreit, Jakob and Jones, Llion and Gomez, Aidan N and Kaiser, {\L}ukasz and Polosukhin, Illia},
  journal={Advances in neural information processing systems},
  volume={30},
  year={2017}
}

@inproceedings{kelley1959critical,
  title={Critical-path planning and scheduling},
  author={Kelley Jr, James E and Walker, Morgan R},
  booktitle={Papers presented at the December 1-3, 1959, eastern joint IRE-AIEE-ACM computer conference},
  pages={160--173},
  year={1959}
}

@inproceedings{narayanan2021memory,
  title={Memory-efficient pipeline-parallel dnn training},
  author={Narayanan, Deepak and Phanishayee, Amar and Shi, Kaiyu and Chen, Xie and Zaharia, Matei},
  booktitle={International Conference on Machine Learning},
  pages={7937--7947},
  year={2021},
  organization={PMLR}
}

@article{yang2021pipemare,
  title={Pipemare: Asynchronous pipeline parallel dnn training},
  author={Yang, Bowen and Zhang, Jian and Li, Jonathan and R{\'e}, Christopher and Aberger, Christopher and De Sa, Christopher},
  journal={Proceedings of Machine Learning and Systems},
  volume={3},
  pages={269--296},
  year={2021}
}

@article{ajanthan2025nesterov,
  title={Nesterov method for asynchronous pipeline parallel optimization},
  author={Ajanthan, Thalaiyasingam and Ramasinghe, Sameera and Zuo, Yan and Avraham, Gil and Long, Alexander},
  journal={arXiv preprint arXiv:2505.01099},
  year={2025}
}

@misc{Karpathy2022,
  author = {Andrej Karpathy},
  title = {\text{NanoGPT}},
  year = {2022},
  publisher = {GitHub},
  journal = {GitHub repository},
  howpublished = {\url{https://github.com/karpathy/nanoGPT}},
  commit = {325be85d9be8c81b436728a420e85796c57dba7e}
}

@article{zhao2024galore,
  title={Galore: Memory-efficient llm training by gradient low-rank projection},
  author={Zhao, Jiawei and Zhang, Zhenyu and Chen, Beidi and Wang, Zhangyang and Anandkumar, Anima and Tian, Yuandong},
  journal={arXiv preprint arXiv:2403.03507},
  year={2024}
}

@article{zhang2024transformers,
  title={Why transformers need adam: A hessian perspective},
  author={Zhang, Yushun and Chen, Congliang and Ding, Tian and Li, Ziniu and Sun, Ruoyu and Luo, Zhi-Quan},
  journal={Advances in neural information processing systems},
  volume={37},
  pages={131786--131823},
  year={2024}
}

@article{dong2025towards,
  title={Towards quantifying the hessian structure of neural networks},
  author={Dong, Zhaorui and Zhang, Yushun and Yao, Jianfeng and Sun, Ruoyu},
  journal={arXiv preprint arXiv:2505.02809},
  year={2025}
}

@article{zhang2015character,
  title={Character-level convolutional networks for text classification},
  author={Zhang, Xiang and Zhao, Junbo and LeCun, Yann},
  journal={Advances in neural information processing systems},
  volume={28},
  year={2015}
}

@inproceedings{gliwa2019samsum,
  title={SAMSum corpus: A human-annotated dialogue dataset for abstractive summarization},
  author={Gliwa, Bogdan and Mochol, Iwona and Biesek, Maciej and Wawer, Aleksander},
  booktitle={Proceedings of the 2nd Workshop on New Frontiers in Summarization},
  pages={70--79},
  year={2019}
}

@inproceedings{saravia2018carer,
  title={CARER: Contextualized affect representations for emotion recognition},
  author={Saravia, Elvis and Liu, Hsien-Chi Toby and Huang, Yen-Hao and Wu, Junlin and Chen, Yi-Shin},
  booktitle={Proceedings of the 2018 conference on empirical methods in natural language processing},
  pages={3687--3697},
  year={2018}
}

@inproceedings{luan2018multi,
  title={Multi-task identification of entities, relations, and coreference for scientific knowledge graph construction},
  author={Luan, Yi and He, Luheng and Ostendorf, Mari and Hajishirzi, Hannaneh},
  booktitle={Proceedings of the 2018 conference on empirical methods in natural language processing},
  pages={3219--3232},
  year={2018}
}

@misc{jordan2024muon,
  author       = {Keller Jordan and Yuchen Jin and Vlado Boza and Jiacheng You and
                  Franz Cesista and Laker Newhouse and Jeremy Bernstein},
  title        = {Muon: An optimizer for hidden layers in neural networks},
  year         = {2024},
  url          = {https://kellerjordan.github.io/posts/muon/}
}

@article{voigt2017eu,
  title={The eu general data protection regulation (gdpr)},
  author={Voigt, Paul and Von dem Bussche, Axel},
  journal={A practical guide, 1st ed., Cham: Springer International Publishing},
  volume={10},
  number={3152676},
  pages={10--5555},
  year={2017},
  publisher={Springer}
}

@article{hui2024qwen2,
  title={Qwen2. 5-coder technical report},
  author={Hui, Binyuan and Yang, Jian and Cui, Zeyu and Yang, Jiaxi and Liu, Dayiheng and Zhang, Lei and Liu, Tianyu and Zhang, Jiajun and Yu, Bowen and Lu, Keming and others},
  journal={arXiv preprint arXiv:2409.12186},
  year={2024}
}

@article{loshchilov2017decoupled,
  title={Decoupled weight decay regularization},
  author={Loshchilov, Ilya and Hutter, Frank},
  journal={arXiv preprint arXiv:1711.05101},
  year={2017}
}

@inproceedings{yu-etal-2018-spider,
    title = "{S}pider: A Large-Scale Human-Labeled Dataset for Complex and Cross-Domain Semantic Parsing and Text-to-{SQL} Task",
    author = "Yu, Tao  and
      Zhang, Rui  and
      Yang, Kai  and
      Yasunaga, Michihiro  and
      Wang, Dongxu  and
      Li, Zifan  and
      Ma, James  and
      Li, Irene  and
      Yao, Qingning  and
      Roman, Shanelle  and
      Zhang, Zilin  and
      Radev, Dragomir",
    editor = "Riloff, Ellen  and
      Chiang, David  and
      Hockenmaier, Julia  and
      Tsujii, Jun{'}ichi",
    booktitle = "Proceedings of the 2018 Conference on Empirical Methods in Natural Language Processing",
    month = oct # "-" # nov,
    year = "2018",
    address = "Brussels, Belgium",
    publisher = "Association for Computational Linguistics",
    url = "https://aclanthology.org/D18-1425",
    doi = "10.18653/v1/D18-1425",
    pages = "3911--3921",
    archivePrefix={arXiv},
    eprint={1809.08887},
    primaryClass={cs.CL},
}

@article{vyas2024soap,
  title={Soap: Improving and stabilizing shampoo using adam},
  author={Vyas, Nikhil and Morwani, Depen and Zhao, Rosie and Kwun, Mujin and Shapira, Itai and Brandfonbrener, David and Janson, Lucas and Kakade, Sham},
  journal={arXiv preprint arXiv:2409.11321},
  year={2024}
}

@inproceedings{thapa2022splitfed,
  title={Splitfed: When federated learning meets split learning},
  author={Thapa, Chandra and Arachchige, Pathum Chamikara Mahawaga and Camtepe, Seyit and Sun, Lichao},
  booktitle={Proceedings of the AAAI conference on artificial intelligence},
  volume={36},
  number={8},
  pages={8485--8493},
  year={2022}
}

@inproceedings{chen2024unveiling,
  title={Unveiling the vulnerability of private fine-tuning in split-based frameworks for large language models: A bidirectionally enhanced attack},
  author={Chen, Guanzhong and Qin, Zhenghan and Yang, Mingxin and Zhou, Yajie and Fan, Tao and Du, Tianyu and Xu, Zenglin},
  booktitle={Proceedings of the 2024 on ACM SIGSAC Conference on Computer and Communications Security},
  pages={2904--2918},
  year={2024}
}

@article{cobbe2021gsm8k,
  title={Training Verifiers to Solve Math Word Problems},
  author={Cobbe, Karl and Kosaraju, Vineet and Bavarian, Mohammad and Chen, Mark and Jun, Heewoo and Kaiser, Lukasz and Plappert, Matthias and Tworek, Jerry and Hilton, Jacob and Nakano, Reiichiro and Hesse, Christopher and Schulman, John},
  journal={arXiv preprint arXiv:2110.14168},
  year={2021}
}

@article{grattafiori2024llama,
  title={The llama 3 herd of models},
  author={Grattafiori, Aaron and Dubey, Abhimanyu and Jauhri, Abhinav and Pandey, Abhinav and Kadian, Abhishek and Al-Dahle, Ahmad and Letman, Aiesha and Mathur, Akhil and Schelten, Alan and Vaughan, Alex and others},
  journal={arXiv preprint arXiv:2407.21783},
  year={2024}
}

\appendix
\section{Convergence Analysis}

To formally bound the asynchronous error caused by the delayed gradients in Phase 2, we first introduce two standard assumptions commonly used in the convergence analysis of non-convex optimization.

\begin{proof}[\textbf{Proof of Lemma \ref{lemma:async_error}}]
Recall from Section \ref{sec:two_phase} that the surrogate gradient used in Phase 2 is defined as
\begin{equation*}
    \widetilde{\lGrad}_l^{(t-\tau)} = \nabla_{\lW_l} L_{\mathrm{local}}\bigl(\lW_l^{(t)}; 
    \lAct_l^{(t-\tau)}, \lGrad_{l+1}^{(t-\tau)}\bigr),
\end{equation*}
where the activations $\lAct_l^{(t-\tau)}$ and its gradients $\lGrad_{l+1}^{(t-\tau)}$ are 
$\tau$ steps stale while the local weight $\lW_l^{(t)}$ is fresh.
In the synchronous case ($\tau = 0$) the stale signals coincide with the current ones:
$\lAct_l^{(t-0)} = \lAct_l^{(t)}$, $\lGrad_{l+1}^{(t-0)} = \lGrad_{l+1}^{(t)}$.
The surrogate gradient therefore reduces to the standard back‑propagation gradient obtained by 
applying the chain rule at the current, globally synchronous state $\lW^{(t)}$:
\begin{equation} \label{eq:chain}
    \frac{\partial \loss(\lW_l)}{\partial \lW_l}
    = \frac{\partial \loss(\lW_l)}{\partial \lAct_{l+1}^{(t)}} 
      \frac{\partial \lAct_{l+1}^{(t)}}{\partial \lW_l}
    = \widetilde{\lGrad}_l^{(t)}.
\end{equation}


In our asynchronous system, however, to break the communication bottleneck, the target block $l$ computes its surrogate local gradient $\widetilde{\lGrad}_l^{(t-\tau)}$ using stale forward activations and stale backward gradients delayed by $\tau$ steps. Thus, activations and gradients is mathematically equivalent to evaluating the loss using the $\tau$-delayed weights of the preceding and succeeding blocks, while using the fresh weights for the local block $l$. 
Naturally, this allows us to rigorously define the surrogate local gradient as the true gradient evaluated on a \textit{mixed-state parameter subspace} $\widetilde{\lW}$ with each block $\widetilde{\lW}_k$:
\begin{equation*}
    \widetilde{\lW}_k = 
    \begin{cases} 
      \lW_l^{(t)}, & \text{if } k = l \text{ (fresh local weights)}\\
      \lW_k^{(t-\tau)}, & \text{if } k \neq l \text{ (stale external signals)}
    \end{cases}
\end{equation*}

Thus, $\widetilde{\lGrad}_l^{(t-\tau)} \equiv \nabla_{\lW_l} \loss(\widetilde{\lW})$. The asynchronous error can be explicitly defined as the distance between the gradient evaluated at the synchronous state $\lW^{(t)}$ and the mixed state $\widetilde{\lW}$. By the Mean Value Theorem and Assumption \ref{assum:cross_hessian}, we replace the maximum Hessian norm with the cross-block Lipschitz constant $H_{l,j}$:
\begin{equation*}
    \left\| \nabla_{\lW_l}\loss(\lW^{(t)}) - \widetilde{\lGrad}_l^{(t-\tau)} \right\| \le \sum_{j=1}^{\bNum} H_{l,j} \left\| \lW_j^{(t)} - \widetilde{\lW}_j \right\|
\end{equation*}

Since the local block $l$ inherently uses its most up-to-date weights during the surrogate local gradient computation, the parameter drift for block $l$ is zero:
\begin{equation*}
    \left\| \lW_l^{(t)} - \widetilde{\lW}_l \right\| = \left\| \lW_l^{(t)} - \lW_l^{(t)} \right\| = 0
\end{equation*}
Consequently, the self-Hessian term $H_{l,l}$ is perfectly eliminated from the summation. The error is solely bounded by the cross-block Hessians ($j \neq l$):
\begin{equation}
    \left\| \nabla_{\lW_l}\loss(\lW^{(t)}) - \widetilde{\lGrad}_l^{(t-\tau)} \right\| \le \sum_{j \neq l} H_{l,j} \left\| \lW_j^{(t)} - \lW_j^{(t-\tau)} \right\| \label{eq:drift_sum}
\end{equation}
Finally, using the triangle inequality and Assumption \ref{assum:cross_hessian}, we bound the accumulated parameter drift over the $\tau$ steps:
\begin{equation*}
    \left\| \lW_j^{(t)} - \lW_j^{(t-\tau)} \right\| \le \sum_{k=1}^{\tau} \left\| \lW_j^{(t-k+1)} - \lW_j^{(t-k)} \right\| \le \lEta \tau G
\end{equation*}
Substituting this back into Eq. (\ref{eq:drift_sum}) and squaring both sides yields the final bound:
\begin{equation*}
    \left\| \nabla_{\lW_l}\loss(\lW^{(t)}) - \widetilde{\lGrad}_l^{(t-\tau)} \right\|^2 \le \left( \lEta \tau G \sum_{j \neq l} H_{l,j} \right)^2
\end{equation*}
This completes the proof.
\end{proof}

\begin{proof}[\textbf{Proof of Theorem \ref{thm:convergence}}]
From the Taylor Expansion result established in Equation \ref{eq:descent}, the loss descent for an individual block $l$ is bounded by:
\begin{equation*}
\loss(\lW_l^{(t+1)}) \le \loss(\lW_l^{(t)}) - \frac{\lEta}{2} \| \nabla_{\lW_l} \loss(\lW_l^{(t)}) \|^2 - \left(\frac{\lEta}{2} - \frac{L_l \lEta^2}{2}\right) \| \tilde{\lGrad}_l^{(t-\tau)} \|^2 + \frac{\lEta}{2} \| \nabla_{\lW_l} \loss(\lW_l^{(t)}) - \tilde{\lGrad}_l^{(t-\tau)} \|^2
\end{equation*}

Given the learning rate condition $\lEta \le \frac{1}{L_l}$, it follows that the coefficient of the third term is non-negative, meaning $\left(\frac{\lEta}{2} - \frac{L_l \lEta^2}{2}\right) \ge 0$. Therefore, we can drop the negative third term on the right-hand side to obtain a relaxed upper bound:
\begin{equation*}
\loss(\lW_l^{(t+1)}) \le \loss(\lW_l^{(t)}) - \frac{\lEta}{2} \| \nabla_{\lW_l} \loss(\lW_l^{(t)}) \|^2 + \frac{\lEta}{2} \| \nabla_{\lW_l} \loss(\lW_l^{(t)}) - \tilde{\lGrad}_l^{(t-\tau)} \|^2
\end{equation*}

Substituting the asynchronous error bound derived in Lemma \ref{lemma:async_error}  into the above inequality, we have:
\begin{equation*}
\loss(\lW_l^{(t+1)}) \le \loss(\lW_l^{(t)}) - \frac{\lEta}{2} \| \nabla_{\lW_l} \loss(\lW_l^{(t)}) \|^2 + \frac{\lEta}{2} \left( \lEta \tau G \sum_{j \neq l} H_{l,j} \right)^2
\end{equation*}

Rearranging the terms to isolate the gradient norm yields:
\begin{equation*}
\frac{\lEta}{2} \| \nabla_{\lW_l} \loss(\lW_l^{(t)}) \|^2 \le \loss(\lW_l^{(t)}) - \loss(\lW_l^{(t+1)}) + \frac{\lEta}{2} \left( \lEta \tau G \sum_{j \neq l} H_{l,j} \right)^2
\end{equation*}

Multiplying both sides by $\frac{2}{\lEta}$ gives:
\begin{equation*}
\| \nabla_{\lW_l} \loss(\lW_l^{(t)}) \|^2 \le \frac{2}{\lEta} [\loss(\lW_l^{(t)}) - \loss(\lW_l^{(t+1)})] + \left( \lEta \tau G \sum_{j \neq l} H_{l,j} \right)^2
\end{equation*}

Summing this inequality over $t = 0$ to $T-1$ and dividing by the total number of steps $T$ leads to a telescoping sum on the loss evaluations:
\begin{equation*}
\frac{1}{T} \sum_{t=0}^{T-1} \| \nabla_{\lW_l} \loss(\lW_l^{(t)}) \|^2 \le \frac{2}{\lEta T} [\loss(\lW_l^{(0)}) - \loss(\lW_l^{(T)})] + \lEta^2 \tau^2 G^2 \left( \sum_{j \neq l} H_{l,j} \right)^2
\end{equation*}

Since the loss is strictly bounded below by $\loss^*$, we know that $\loss(\lW_l^{(0)}) - \loss(\lW_l^{(T)}) \le \loss(\lW_l^{(0)}) - \loss^*$. Substituting this absolute bound gives our final result:
\begin{equation*}
\frac{1}{T} \sum_{t=0}^{T-1} \| \nabla_{\lW_l} \loss(\lW^{(t)}) \|^2 \le \frac{2[\loss(\lW_l^{(0)}) - \loss^*]}{\lEta T} + \lEta^2 \tau^2 G^2 \left( \sum_{j \neq l} H_{l,j} \right)^2
\end{equation*}
\end{proof}
\section{Implementation Details}
\label{sec:appendix_impl}
\subsection{Model and Split Configuration}
We fine-tune the Qwen2.5-1.5B-Instruct model \cite{hui2024qwen2}.  
The U‑shaped split architecture is denoted by a triplet $[n_{\text{head}}, n_{\text{bulk}}, n_{\text{tail}}]$ giving the number of Transformer blocks on the client head, server bulk, and client tail, respectively.  
Throughout Section~5.2 we use the default split $[3,22,3]$.  
All methods employ a batch size of 16 with a cosine learning rate schedule.
\subsection{Optimizer Details}
\textbf{Synchronous AdamW (baseline).}  
We use the standard PyTorch implementation of AdamW \cite{loshchilov2017decoupled} with 
$\beta_1{=}0.9,\beta_2{=}0.999$, weight decay $0.09$, and $\epsilon{=}1\times10^{-8}$.
The base learning rate is set to $\eta{=}2\times10^{-5}$ in the main experiments.  

\textbf{CommuteProp and CommuteProp (NS).}  
CommuteProp wraps any base optimizer; the same AdamW hyperparameters serve the CommuteProp~(AdamW) variant.  
The staleness delay $\tau$ is emulated via a queue‑based simulator: a gradient update is triggered only after the corresponding forward activation has been queued for exactly $\tau$ micro‑batches.  
For the utility comparison (Table \ref{tab:utl}) we set $\tau{=}16$. The ``NS'' variant applies the Kronecker‑Block Newton–Schulz preconditioner (Algorithm~\ref{alg:ns_precond}) before the AdamW step; we use a block size $b{=}64$ and refresh the preconditioner every 200 steps.

\textbf{Nesterov Async.}  
We re‑implement the asynchronous pipeline parallel method of \cite{ajanthan2025nesterov} with its default momentum factor $0.9$ and Nesterov correction; the staleness delay is fixed at $\tau{=}16$ for the fairness comparison.

\textbf{Synchronous Muon (baseline).}  
We use the standard PyTorch implementation of Muon \cite{jordan2024muon} with 
$\beta_1{=}0.9,\beta_2{=}0.999$, weight decay $0.09$, and $\epsilon{=}1\times10^{-8}$.
The base learning rate is set to $\eta{=}2\times10^{-5}$ in the main experiments.  

\textbf{CommuteProp (Muon).}  
The ``Muon'' variant replaces AdamW with the Muon optimizer \cite{jordan2024muon} (momentum $0.95$, Nesterov acceleration enabled) while keeping the CommuteProp two‑phase pipeline unchanged.  

\subsection{Full-Training Comparison and Iso‑time Comparison}

\label{sec:iso_time}

In our experiments we report both full‑training and iso‑time comparisons. The full‑training comparison verifies the utility of CommuteProp in end‑to‑end training, while the iso‑time comparison, controlled by equal wall‑clock budgets, directly captures the efficiency gain in the split learning scenario.

 Our experiment use one single NVIDIA RTX A6000. For the iso‑time comparison, we first profile the per‑stage computation costs. We record the forward and backward times of the synchronous baseline (AdamW) and of CommuteProp (NS) over the first 50 training iterations and average the measurements. To isolate the speedup coming purely from computation‑communication overlap, we exclude the time of optimizer updates, gradient zeroing, and learning rate scheduling, as these operations are nearly identical in both settings, non‑overlappable, and would only add a constant offset to the end‑to‑end latency without affecting the relative gain from pipelining. Denote the $l=1$ as the client node, $l=2$ as the server node. The resulting mean per-stage costs are:

\begin{table}[htbp]
  \centering
  \small
  \captionsetup{font=small}
  \caption{Per‑stage computation costs profiled over the first 50 iterations.}
  \label{tab:stage_cost}
  \resizebox{\textwidth}{!}{
  \begin{tabular}{@{}lll@{}}
    \toprule
    Method & Stage cost (s) & Notes \\
    \midrule
    Sync baseline & $t_{1,s}+t_{2,s} =0.089$ & unified forward+backward \\
    \multirow{2}{*}{CommuteProp (NS)} 
      & Phase~1: $t_{1,1},t_{2,1}=[0.0042,\,0.0186]$ & \\
      & Phase~2: $t_{1,2},t_{2,2}=[0.0364,\,0.0695]$ & separation of computation and update \\
    \bottomrule
  \end{tabular}
  }
\end{table}

By using the Equation \ref{equ:efficiency} and assuming one-way communication time is $0.05$, we have $\kappa_a/\kappa_s=1.92$; under an equal wall‑clock budget the asynchronous method can process $1.92\times$ more data. Thus, under the iso‑time constraint, the synchronous epoch budget is scaled by the inverse of the speedup. For example, if CommuteProp trains $5$ epochs, under the time-constraint scenario, it trains for only $5/1.92 \approx 2.6$ epochs.

\subsection{Client-server deployment experiment}
\label{sec:real_exp}

To validate CommuteProp beyond the queue-based simulation, we conduct real distributed client--server experiments on two GPUs and use Linux \texttt{tc} to control the bandwidth (1 Gbps) and latency. Table~\ref{tab:real_deploy} reports the per-iteration time of the synchronous baseline and CommuteProp under different latency settings. The measured speedups are consistently higher than the simulated estimates, confirming that the bubble-free synchronous simulation provides a conservative evaluation of CommuteProp's practical advantage.

\begin{table}[H]
  \centering
  \footnotesize
  \captionsetup{font=footnotesize, labelfont=bf}
  \caption{Real two-GPU client--server experiments under 1 Gbps bandwidth, with latency controlled by Linux \texttt{tc}.}
  \label{tab:real_deploy}
  \begin{tabular}{ccccc}
    \toprule
    \textbf{Configured latency} & \textbf{Measured latency} & \textbf{Sync/iter (s)} & \textbf{Async/iter (s)} & \textbf{Speedup} \\
    \midrule
    10 ms  & 10.9 ms  & 0.155 & 0.068 & 2.28$\times$ \\
    50 ms  & 51.0 ms  & 0.325 & 0.067 & 4.89$\times$ \\
    100 ms & 101.2 ms & 0.534 & 0.066 & 8.07$\times$ \\
    200 ms & 201.4 ms & 0.932 & 0.064 & 14.55$\times$ \\
    \bottomrule
  \end{tabular}
\end{table}

To further benchmark CommuteProp against stronger system and scheduling baselines, we additionally include PipeDream-2BW, a weight-stashing pipeline method, and evaluate all methods under the same 15-minute wall-clock budget; the classification and generation results are reported in Table~\ref{tab:wallclock}. CommuteProp achieves competitive or better performance than PipeDream-2BW under the same wall-clock budget, particularly when $ \tau=4$. We initially selected Nesterov Async as the primary asynchronous baseline because it avoids multi-version weight stashing and is therefore better suited to memory-constrained split-learning clients.

\begin{table}[htbp]
  \centering
  \footnotesize
  \captionsetup{font=footnotesize, labelfont=bf}
  \caption{Model utility under the 15-minute equal wall-clock budget. For each dataset--metric pair, the two columns correspond to $ \tau=4$ and $ \tau=16$.}
  \label{tab:wallclock}
  \setlength{\tabcolsep}{0pt}
  \begin{tabular*}{\textwidth}{@{\extracolsep{\fill}} l cccccccccc}
    \toprule
    & \multicolumn{4}{c}{\textbf{Classification (Acc)}} & \multicolumn{6}{c}{\textbf{Generation}} \\
    \cmidrule(lr){2-5} \cmidrule(lr){6-11}
    & \multicolumn{2}{c}{\textbf{Emotions}} & \multicolumn{2}{c}{\textbf{AG-News}} & \multicolumn{4}{c}{\textbf{SAMSum}} & \multicolumn{2}{c}{\textbf{SciERC}} \\
    \cmidrule(lr){2-3} \cmidrule(lr){4-5} \cmidrule(lr){6-9} \cmidrule(lr){10-11}
    & \multicolumn{2}{c}{\textbf{Acc}}  & \multicolumn{2}{c}{\textbf{Acc}} & \multicolumn{2}{c}{\textbf{ROUGE-L}} & \multicolumn{2}{c}{\textbf{METEOR}} & \multicolumn{2}{c}{\textbf{F1}} \\
    \cmidrule(lr){2-3} \cmidrule(lr){4-5} \cmidrule(lr){6-7} \cmidrule(lr){8-9} \cmidrule(lr){10-11}
    \textbf{Algorithm} & $ \tau=4$  & $\tau=16$ & $ \tau=4$  & $\tau=16$ & $ \tau=4$ & $\tau=16$ & $ \tau=4$  & $\tau=16$ & $ \tau=4$ &  $\tau=16$ \\
    \midrule
    Sync                & 0.7998 & 0.7378 & 0.8994 & 0.9120 & 0.3576 & 0.3585 & 0.3052 & 0.3042 & 0.3915 & 0.3805 \\
    CommuteProp         & 0.9076 & 0.9189 & 0.9220 & 0.9320 & 0.3704 & 0.3743 & 0.3203 & 0.3196 & 0.7747 & 0.7591 \\
    CommuteProp (NS)    & 0.8973 & 0.9097 & 0.9271 & 0.9300 & 0.3698 & 0.3707 & 0.3213 & 0.3183 & 0.7596 & 0.7814 \\
    CommuteProp (Muon)  & 0.9189 & 0.9251 & 0.8960 & 0.9120 & 0.3788 & 0.3999 & 0.3199 & 0.3402 & 0.7199 & 0.7868 \\
    PipeDream-2BW       & 0.9076 & 0.8819 & 0.9340 & 0.9060 & 0.3694 & 0.3697 & 0.3202 & 0.3159 & 0.7545 & 0.6743 \\
    Nesterov Async      & 0.6653 & 0.6181 & 0.8820 & 0.8800 & 0.3396 & 0.3323 & 0.2882 & 0.2782 & 0.0597 & 0.0623 \\
    \bottomrule
  \end{tabular*}
\end{table}

\subsection{Ablation Setup}
The ablation study (Section~5.3, Figure~4) sweeps three factors while keeping all other settings at their default values:
\begin{enumerate}
    \item \textbf{Staleness delay $\tau$}: $\{4, 8, 16\}$ steps.
    \item \textbf{Learning rate $\eta$}: $\{1\times10^{-5}, 2\times10^{-5}, 5\times10^{-5}\}$.
    \item \textbf{Split architecture}: $[1,26,1]$, $[3,22,3]$
\end{enumerate}
When not varied, the default values are $\tau{=}16$, $\eta{=}2\times10^{-5}$, and the $[3,22,3]$ split.  
Table~\ref{tab:ablation_grid} summarises the parameter combinations evaluated; any run not listed inherits the defaults described above.

\begin{table}[h]
  \centering
  \small
  \captionsetup{font=small}
  \caption{Grid of ablation settings.  The default configuration ($\tau{=}16$, $\eta{=}2{\times}10^{-5}$, $[3,22,3]$) is highlighted and serves as the anchor for all single‑factor sweeps.}
  \label{tab:ablation_grid}
  \begin{tabular}{@{}llll@{}}
    \toprule
    Varied factor & Fixed factors & Values explored \\
    \midrule
    Staleness $\tau$ & $\eta{=}2{\times}10^{-5}$, split $[3,22,3]$ & 4, 8, \textbf{16} \\
    Learning rate $\eta$ & $\tau{=}16$, split $[3,22,3]$ & $1{\times}10^{-5}$, $\mathbf{2{\times}10^{-5}}$, $5{\times}10^{-5}$ \\
    Architecture & $\tau{=}16$, $\eta{=}2{\times}10^{-5}$ & $[1,26,1]$, $\mathbf{[3,22,3]}$ \\
    \bottomrule
  \end{tabular}
\end{table}
\subsection{Datasets and Evaluation}

We evaluate on four diverse NLP benchmarks: Emotion, AG-News, SAMSum, SciERC and Spider.  
Emotion~\cite{saravia2018carer} is a 6-class emotion classification dataset (sadness, joy, love, anger, fear, surprise) with 16,000 training samples and 2,000 test samples; texts are short social media posts with an average length of approximately 20 words.  
AG News~\cite{zhang2015character} is a 4-class topic classification dataset (World, Sports, Business, Sci/Tech) with 120,000 training samples and 7,600 test samples; texts are news articles averaging 40--50 words.  
SAMSum~\cite{gliwa2019samsum} is a dialogue summarization dataset with 14,732 training samples and 819 test samples; each dialogue contains multiple turns of casual conversation, and the reference summaries average approximately 20 words.  
SciERC~\cite{luan2018multi} is a scientific relation extraction dataset with 3,219 training samples and 974 test samples; texts are scientific paper abstracts, and the task is to classify the relation between annotated entity pairs into one of 7 relation types (labels 0--6).
Spider~\cite{yu-etal-2018-spider} is a cross-domain text-to-SQL dataset with 7,000 training samples and 1,034 test samples; each example provides a natural language question paired with the corresponding SQL query spanning 200 databases.
For all datasets, we report final results on the official test split.  
Emotion, AG-News and Spider are trained for 3 epochs, while SAMSum and SciERC are trained for 5 epochs. 

For preprocessing, we apply the Qwen2.5-1.5B-Instruct chat template  to format each example into a consistent instruction-following format, and truncate all sequences to a maximum length of 2,048 tokens. For Spider, the instruction asks the model to translate the question into an SQL query.  
Padding is performed dynamically per batch, which also masks the user-instruction portion of each sequence in the loss computation using \texttt{labels=-100}, so that the model is trained only on the assistant response tokens.

Reported metrics are Accuracy (for classification), ROUGE-L and METEOR (for summarization), macro-F1 (for relation extraction), and exact-set-match (EM) accuracy (for text-to-SQL), consistent with prior work.  

For Emotion (6-class emotion classification) and AG-News (4-class topic classification), we compute accuracy as the fraction of samples where the model's predicted label exactly matches the ground-truth label.  
For SAMSum (dialogue summarization), we report ROUGE-L F-measure as the primary metric and METEOR as a secondary metric; both are computed using the HuggingFace \texttt{evaluate} library with default settings (ROUGE-L with stemming enabled via \texttt{use\_stemmer=True}).  
For SciERC (scientific relation extraction with 7 relation types), we report macro-averaged F1 score across all 7 classes.
For Spider, we report exact-set-match accuracy, which compares the execution results of the predicted SQL query against the gold query on the target database—a prediction is considered correct only if the two result sets are identical.

During inference, we use greedy decoding (\texttt{do\_sample=False}) with dataset-specific maximum generation lengths: 10 tokens for Emotion and AG-News (single-label output), 20 tokens for SciERC (single-digit output), 128 tokens for Spider (SQL query generation), and 128 tokens for SAMSum (summary generation).

\subsection{Integration with DualGuard}
\subsubsection{Introductions}
In Section \ref{sec:pp}. We integrate CommuteProp into the DualGuard framework~\cite{liu2025dualguard}.  DualGuard assumes that the adversaries can access to (1) the pre-trained model weights (publicly available), (2) the transmitted smashed data and gradients during training, and (3) an auxiliary dataset for training attack models. The attacker aims to reconstruct the private training data. Specifically, they introduce four attack mechanisms:

\textbf{SIP(Smashed-data Inversion based on Pre-trained model)}: The attacker trains an inversion model (a GRU-based decoder followed by an MLP projection) on an auxiliary dataset using the pre-trained head model. The inversion model learns to map smashed data (head model activations) back to input token sequences.

\textbf{TAG (Text Revealer from Gradients)}:  TAG reconstructs output labels by matching gradients. The attacker intercepts the gradients flowing back from the tail model to the server. Using the pre-trained tail model, the attacker initializes dummy label logits and iteratively optimizes them to minimize the distance between the dummy gradients (produced by the pre-trained tail on the dummy labels) and the true gradients. The recovered labels in LLM tasks substantially overlap with the input, enabling reconstruction of the private sequence.

\textbf{LAMP (Label-extracting Attack with Model Priors)}: LAMP extends TAG by incorporating stronger language model priors. Instead of optimizing over continuous logit space, LAMP alternates between (1) discrete token selection based on the current dummy logits and (2) gradient matching optimization.

\textbf{CPTA (Connecting Pre-trained Tail Attack)}: CPTA directly connects the pre-trained tail model to the server's output hidden states, bypassing the trained tail model entirely. The attacker simply feeds the intercepted server outputs into the publicly available pre-trained tail model and evaluates its generation against the private labels. 

\textbf{DualGuard Defense.} DualGuard defends against these attacks through 
two stages:

\textit{Stage~1: Local Warm-Up Parameter Space Transformation.} Before 
formal fine-tuning, a three-layer projection MLP is inserted between the head and 
tail models. The head, tail, and ProjectionMLP are jointly optimized on the client 
side via multi-task learning with three objectives: (1)~maximizing the SIP inversion 
loss to disrupt smashed-data reconstruction, (2)~maximizing the pre-trained tail's 
cross-entropy loss to deactivate gradient matching, and (3)~minimizing the standard 
LM loss to preserve task performance. LoRA adapters are used for 
parameter-efficient training on the resource-constrained client. The trunk model is 
not involved in this stage.

\textit{Stage~2: Global Fine-Tuned Parameter Space Retention.} During 
formal split-based fine-tuning, the head model is frozen to maintain the secure 
forward activations established in Stage~1. LoRA adapters on the server (trunk) and 
tail models are fine-tuned on the target dataset, while base model weights remain 
frozen. An anti-pre-trained-tail loss term ($\lambda_3 / \mathcal{L}_{\text{global-pt}}$) 
is added to prevent the tail from reverting to its vulnerable pre-trained parameter 
space. In our experiments, we replace the original synchronous training algorithm 
in Stage~2 with our CommuteProp algorithm.

\subsubsection{HyperParameters}

Our hyper-parameters are nearly the same as DualGuard. For split architecture, the original paper claims the $[2,12,2]$ for Qwen2-1.5B, which does not match the total block numbers(28). Thus, we replace it with [3,22,3]. Full hyper-parameters are detailed in Table \ref{tab:dualguard-hp}.

\begin{table}[ht]
\footnotesize
\centering
\caption{Hyperparameter configuration for DualGuard + CommuteProp experiments.}
\label{tab:dualguard-hp}
\begin{tabular}{@{}l|l@{}}
\toprule
\textbf{Setting} & \multicolumn{1}{c}{\textbf{Hyperparameters}} \\
\midrule
Model \& Data & base = Qwen2-1.5B, split = [3,22,3], max\_seq\_len = 512, dataset = GSM8k\cite{cobbe2021gsm8k} \\
LoRA & $r = 2$, $\alpha = 32$, dropout = 0, targets = \{q\_proj, v\_proj\} \\
\midrule
Warm-Up & epoch = 2, bs = 4, lr = $2\times10^{-5}$, AdamW, wd = 0.01, $\lambda_1 = 40$, $\lambda_2 = 80$ \\
Fine-Tuning & epoch = 4, bs = 4, lr = $2\times10^{-5}$, AdamW, wd = 0.01, $\lambda_3 = 10$, $\tau = 4$ \\
\midrule
Attack Eval. & samples = \{SIP: 131, TAG/LAMP: 40\}, iters = \{TAG: 500, LAMP: 200 ($\times$30 inner)\}, PT-tail wt = 70.0 \\
\bottomrule
\end{tabular}
\end{table}
\subsubsection{Experiment Result}
\label{sec:more_result}
In the \textit{LoRA} setting, LoRA adapters ($r{=}2$) are 
injected into all model segments; only adapter weights on the server and tail  are updated while base parameters remain frozen. In the \textit{full-parameter} setting, the server is trained without LoRA---all base weights are updated---while head and tail remain LoRA-based from Stage~1. 
\begin{table}[htbp]
  \centering
  \small
  \renewcommand{\arraystretch}{0.85}
  \caption{Performance comparison under different attack methods.}
  \label{tab:attack_performance}
  \begin{tabular*}{\textwidth}{@{\extracolsep{\fill}}lcccccc}
    \toprule
    & \multicolumn{2}{c}{\textbf{Performance ($\uparrow$)}} & \multicolumn{4}{c}{\textbf{Attack (Rouge-L-F \%, $\downarrow$)}} \\
    \cmidrule(lr){2-3} \cmidrule(lr){4-7}
    Method & Rouge-L & METEOR & SIP & TAG & LAMP & CPTA \\
    \midrule
    w/o defense (naive USL)       & 86.23 & 88.48 & 86.87 & 65.70 & 63.13 & 17.75 \\
    Full-param DualGuard (sync)   & 85.76 & 88.13 & 5.64  & 38.67 & 38.26 & 0.53  \\
    Full-param CommuteProp ($\tau{=}4$) & 84.93 & 87.51 & 5.64 & 9.45  & 11.54 & 0.00 \\
    Full-param CommuteProp(NS)    & 84.89 & 87.25 & 5.64  & 23.15 & 23.21 & 5.43  \\
    LoRA DualGuard (sync)& 85.12 & 86.92 & 5.64 & 5.06 & 5.81 & 3.54 \\
    LoRA CommuteProp ($\tau{=}4$) & 85.18 & 87.13 &  5.64 & 4.34 & 4.57 & 3.58 \\
    
    LoRA CommuteProp(NS)  & 85.29 & 87.17 & 5.64  & 6.66  & 6.20  & 6.76  \\
   
    \bottomrule
  \end{tabular*}
\end{table}

The results indicate that CommuteProp achieves stronger privacy protection than synchronous DualGuard with negligible utility degradation. We attribute this improvement to the implicit 
regularization induced by stale activations: the temporal offset  between the head output and the server's computation introduces a controlled  distributional shift, further decorrelating the server's output from the pre-trained model's expectations and thereby weakening gradient-matching attacks.

\section{Additional Results}
\label{sec:addition_basemodel}
\subsection{Addition Experiment Result on Different Base-models}

\textbf{Setup} We additionally run experiments with a larger base model, Llama3.2-3B-Instruct\cite{grattafiori2024llama}, keeping the same U‑shaped split [3, 22, 3] and all hyper‑parameters from Appendix B (batch size 16, cosine learning rate schedule, etc.). To stabilise asynchronous training on this larger model, we introduce gradient clipping with a maximum global norm of 1.0. Using the same profiled one‑way communication time and technique in Section \ref{sec:iso_time}, we obtain a speed‑up factor of $1.98\times$.

\begin{table}[htbp]
  \centering
  \small
  \captionsetup{font=small}
  \caption{Per‑stage computation costs profiled over the first 50 iterations For Llama3.2-3B model}
  \label{tab:stage_cost}
  \begin{adjustbox}{width=\linewidth, center}
  \begin{tabular}{@{}lll@{}}
    \toprule
    Method & Stage cost (s) & Notes \\
    \midrule
    Sync baseline & $t_{1,s}+t_{2,s}=0.107$ & unified forward+backward \\
    \multirow{2}{*}{CommuteProp (NS)} 
      & Phase~1: $t_{1,1}+t_{2,1}=[0.0041,0.0207]$ & \\
      & Phase~2: $t_{1,2},t_{2,2}=[0.0406,\,0.0701]$ & separation of computation and update \\
    \bottomrule
  \end{tabular}
  \end{adjustbox}
\end{table}

\textbf{Results} Due to limited computational resources and time constraints, we were unable to exhaustively evaluate the full grid of baselines and tasks. For the experiments we conducted, we set the number of training epochs to 1 for AGNews. The results, reported in Table \ref{tab:llama}, — consistently demonstrates that CommuteProp (with NS preconditioner) preserves utility close to the synchronous AdamW baseline while effectively hiding communication latency.

\begin{table}[H]
  \vspace{-0.5cm}
  \centering
  \small
  \captionsetup{font=small}
  \renewcommand{\arraystretch}{0.8}
  \caption{Utility comparison under full-training and Equal time budgets}
  \label{tab:llama}
  \begin{tabular*}{\textwidth}{@{\extracolsep{\fill}}lccc}
    \toprule
    \textbf{Algorithm} & \textbf{Emotion (Acc)} & \textbf{AGNews (Acc)} & \textbf{SciERC (F1 Score)} \\
    \midrule
    Zero-shot                & 40.8  & 68.00    & 4.55   \\
    AdamW                    & 93.75 & 96.50   & 82.94   \\
    Nadam                    & 65.25 & 88.12   & 12.84    \\
    AdamW (Time constraint)  & \textbf{93.65} & 95.00   & 76.84   \\
    CommuteProp (NS)         & 93.41 & \textbf{95.21}   & \textbf{83.40}\\
    \bottomrule
  \end{tabular*}
\end{table}


\subsection{Empirical Validation of Cross-Block Hessian Structure}
\label{sec:hessian}

To empirically examine Assumption~\ref{assum:cross_hessian}, we directly measure the Hessian statistics of the evaluated split architectures, including Qwen2.5-1.5B. Table~\ref{tab:hessian} reports the element-wise statistics of the diagonal (within-block) and cross-block Hessian entries. The results show that cross-block interactions are substantially weaker than the dominant within-block terms, supporting the hierarchical smoothness assumption used in our analysis. This observation is also consistent with prior studies on the approximately block-diagonal Hessian structure of large neural networks \cite{zhang2024transformers,dong2025towards}.

Following the empirical protocols in \cite{zhang2024transformers,dong2025towards}, we use element-wise Hessian statistics to support the structural assumption rather than directly estimating the block operator norm, whose computation is prohibitive at the billion-parameter scale. We emphasize that our convergence and performance results are obtained from full-training runs and are independent of this estimate.

\begin{table}[htbp]
  \centering
  \footnotesize
  \captionsetup{font=footnotesize, labelfont=bf}
  \caption{Measured element-wise Hessian statistics on the evaluated Qwen2.5-1.5B split architectures.}
  \label{tab:hessian}
  \begin{tabular}{lccc}
    \toprule
    \textbf{Statistic} & \textbf{Diagonal} & \textbf{Cross-block} & \textbf{Cross/Diagonal} \\
    \midrule
    Max    & $1.49 \times 10^{-2}$ & $7.72 \times 10^{-3}$ & 0.39 \\
    P99    & $6.09 \times 10^{-3}$ & $3.11 \times 10^{-4}$ & 0.05 \\
    Median & $4.05 \times 10^{-5}$ & $2.81 \times 10^{-6}$ & 0.07 \\
    \bottomrule
  \end{tabular}
\end{table}


\end{document}